%% file: main.tex
\input{Tex/head}

\begin{document}

\title{
\sysname: Efficient and Effective Attack String Generation for Regular Expression Denial of Service
Vulnerabilities
}

\author{
Shangzhi Xu$^{\dagger\S}$,~Ziqi Ding$^\dagger$,~Xiao Cheng$^{\ddagger*}$,~Yuekang Li$^\dagger$,~Nan Sun$^\dagger$,~Benjamin Turnbull$^\dagger$,\\~Shuangxiang Kan$^\dagger$,~Siqi Ma$^{\P*1}$\\
{\small $^\dagger$The University of New South Wales,~$^\ddagger$Macquarie University,~$^\S$CSIRO,~$^\P$The University of Wollongong}\\
{\small $^*$Corresponding Authors}\\
{\small \{z5500277, ziqi.ding1, yuekang.li, nan.sun, benjamin.turnbull, shuangxiangkan\}@unsw.edu.au}
{\small xiao.cheng@mq.edu.au,~siqim@uow.edu.au}
}


\maketitle
\footnotetext[1]{S. Ma is partially funded by the Australian Research Council (ARC) Future Fellowship (FT - 250100652)}
\begin{abstract}
Regular Expression Denial-of-Service (ReDoS) attacks constitute a critical class of resource-exhaustion vulnerabilities. In such attacks, 
    adversaries exploit the pathological worst-case execution behavior of regular expression (regex) engines to induce highly asymmetric computational workloads, 
    ultimately exhausting system resources and degrading service availability.
To protect systems against ReDoS attacks, 
    numerous detection techniques have been proposed that simulate the attack process by generating attack strings to proactively exploit ReDoS vulnerabilities at the early development stage and facilitate remediation.
Existing techniques broadly fall into two classes: static analyses that search for pathological regex structures, and dynamic exploration methods that synthesize candidate attack strings.
However, 
    the generated attack strings are often impractical for real-world exploitation because they usually assume unrealistic input-length budgets and do not validate the effectiveness and efficiency of the attack at the program level. 
Therefore, 
    many generated strings fail to trigger vulnerable regexes when applied to real-world programs, 
    further limiting the practical utility.

To address these shortcomings,
    we introduce an effective
and efficient attack string generator, \sysname, designed to
synthesize attack inputs that are both feasible within realistic
length budgets and validated at the program level, enabling
effective exploitation of ReDoS vulnerabilities in real-world
programs. Specifically, we first define three vulnerable patterns
based on our observation and formal verification. According
to the patterns, \sysname conducts a synthesis technique
to generate attack strings, and then refines and validates the
strings with ReDoS-specific compositional concolic execution to
guarantee real-world exploitability. We evaluated \sysname
on a baseline corpus of 17,962 real-world vulnerable regexes.
Compared to state-of-the-art tool RENGAR, \sysname produced
attack strings that achieve efficiency gains ranging from
97.2× to 3872.4×, thereby significantly enhancing practical
exploitability assessment in real-world contexts. \sysname
also discovered 59 exploitable ReDoS instances across 12
popular projects, which is 25 more than RENGAR.
\end{abstract}


%
\IEEEpeerreviewmaketitle

\section{Introduction}
\label{sec:introduction}
\input{Tex/intro}

\vspace{-10pt}
\section{Background}
\label{sec:background}
\input{Tex/background}

\section{Formal Basis for Attack String Generation}
\label{sec:study}
\input{Tex/study}

\section{\sysname}
\label{sec:pufferdos}
\input{Tex/pufferdos}

\section{Evaluation}
\label{sec:experiment}
\input{Tex/experiments}

\section{Related Work}
\label{sec:related}
\input{Tex/relatedWork}


\section{Conclusion}
\label{sec:conclusion}
\input{Tex/conclusion}

\section{Ethics considerations}
\noindent
This work does not involve human subjects, and no personal or otherwise sensitive data were collected, created, or processed. 
This work raises no ethical concerns. All testing was conducted in a local offline environment with no interaction with real-world systems or user data. All previously undisclosed vulnerabilities were responsibly disclosed to developers through a coordinated process, including technical details and proof-of-concept inputs, with sufficient time allowed for patching before public discussion.





\bibliographystyle{Format/IEEEtranS}
\bibliography{reference}
%



\appendices
\section{}
\label{sec:appendix}

\input{Tex/appendices}
\end{document}

%% file: Tex/head.tex
\documentclass[conference,compsoc]{IEEEtran}
\usepackage{enumitem}
\usepackage{tabularx}
\usepackage{graphicx}
\usepackage{diagbox}
\usepackage{tikz}
\usetikzlibrary{positioning, matrix, shapes, backgrounds}

\usepackage[ruled,vlined,linesnumbered]{algorithm2e}
\usepackage{mdframed}
\usepackage{xcolor}
\usepackage{tcolorbox}
\usepackage{makecell}
\usepackage{listings}
\usepackage{array}     
\usepackage{amsmath}
\usepackage{caption}   
\usepackage{subcaption}
\usepackage{tikz}  
\usepackage{amsthm}
\usepackage{lipsum}  
\usepackage{booktabs}
\usepackage{multirow}
\usepackage{pifont}
\usepackage{colortbl}
\usepackage{xspace}
\usepackage{pifont}
\usepackage{threeparttable}
\usepackage{adjustbox}
\usepackage{tikz}
\usepackage{pgf}
\usepackage{tcolorbox}
\usepackage{multicol}
\usepackage{hyperref}
\usepackage{enumitem}
\usepackage{tabularx}
\usepackage{graphicx}
\usepackage{mdframed}
\usepackage{xcolor}
\usepackage{tcolorbox}
\usepackage{makecell}
\usepackage{listings}
\usepackage{array}     
\usepackage{afterpage}
\usepackage{float}
\usepackage{caption}   
\usepackage{subcaption}
\usepackage{tikz}  
\usepackage{lipsum}  
\usepackage{booktabs}
\usepackage{multirow}
\usepackage{colortbl}
\usepackage{xspace}
\usepackage{pifont}
\usepackage{threeparttable}
\usepackage{adjustbox}
\usepackage{tikz}
\usepackage{pgf}
\usepackage{tcolorbox}
\usepackage{multicol}
\usepackage{hyperref}
\usepackage{amsthm}
\usepackage{amssymb}
\newtheorem{definition}{Definition}
\newtheorem{pattern}{Pattern}

\newtheorem{theorem}{Theorem}[section] 
\newtheorem{lemma}[theorem]{Lemma}

\newtheorem{assumption}[theorem]{Heuristic} 
\ifCLASSOPTIONcompsoc
  \usepackage[nocompress]{cite}
\else
  \usepackage{cite}
\fi

\ifCLASSINFOpdf
\else
\fi

\newcommand{\sysname}  {\textsc{PufferDoS}\xspace}

\newcommand{\tool}[1]{{\small{\textsf{{#1}}}}}
\newcommand{\func}[1]{{\texttt{{#1}}}}

\definecolor{myblue}{RGB}{70,130,180}
\definecolor{darkred}{RGB}{139, 0, 0}

\definecolor{commentgreen}{RGB}{2,112,10}
\definecolor{eminence}{RGB}{108,48,130}
\definecolor{weborange}{RGB}{255,165,0}
\definecolor{frenchplum}{RGB}{129,20,83}



%% file: Tex/intro.tex
Regular expression Denial of Service (ReDoS) is  an attack that occurs when pathological regular expression (regex) constructs force a backtracking engine to explore a superlinear or exponential number of matching paths, enabling crafted inputs to exhaust CPU resources and disrupt services.
Its severity is amplified by the pervasive use of regexes across diverse domains of computer science,   including programming languages, 
    file and data processing, 
    database querying, and large language model pipelines~\cite{bartoli2016inference,davis2018impact,siddiq2024understanding}.
Recent empirical studies show that
    ReDoS ranks as the fourth most prevalent vulnerability in the JavaScript \tool{npm} ecosystem and the sixth in the Python \tool{PyPI} ecosystem~\cite{hasan2025model}.
According to Snyk Security Research Team's report~\cite{snyk2020state},
    the number of disclosed ReDoS vulnerabilities has been steadily increasing, with a notable 143\% surge observed in 2018.
On the Common Vulnerabilities and Exposures (CVE) database,
    more than 350 ReDoS vulnerabilities have been disclosed since 2016~\cite{cve,bhuiyan2025sok}.
Real-world incidents further demonstrate its impact: in June 2016, a catastrophic-backtracking regex caused a 34-minute outage on Stack Overflow, and in July 2019, a faulty WAF regex triggered a 27-minute global outage at Cloudflare~\cite{stackoverflow2016outage, cloudflare2019outage}.

Numerous ReDoS detection and attack string generation tools have been developed to uncover vulnerable regexes and produce corresponding proof-of-concept attack inputs,      
    enabling developers to recognize and address potential vulnerabilities early in the development cycle.
Existing methods can be broadly categorized into static analysis-guided, dynamic exploration-based, and execution-validation-based (hybrid) approaches.
Static analysis-guided methods~\cite{doyensec_regexploit_2021,kirrage2013static,parolini2023sound,rathnayake2014static} identify ReDoS vulnerabilities by matching expert-defined regex patterns and constructing corresponding attack strings based on predefined templates.
dynamic exploration-based techniques~\cite{huang2025towards,liu2024co3,kim2019target} leverage fuzzing or guided input exploration to discover inputs that cause matching times to exceed a predefined threshold.
Recent approaches integrate execution validation,   
    combining pattern-guided analysis with dynamic verification to enhance accuracy~\cite{davis2018impact,wang2023effective,noller2018badger}.
Moreover,
    for other vulnerability classes beyond ReDoS,
    concolic execution is widely employed to ensure real-world exploitability~\cite{huang2025towards,liu2024co3,kim2019target}.

However,
    despite the prevalence of ReDoS and recent progress in detection or attack generation,
    prior studies have shown that developers often perceive ReDoS vulnerabilities as low-severity issues and tend to ignore them rather than applying fixes.
The main reasons are twofold: (1) the crafted attack strings to exploit ReDoS are typically unrealistically long,
    rendering such attacks impractical under most real-world deployment settings~\cite{bhuiyan2025sok};
    and (2) as noted in recent discussions~\cite{mlsecopsredospodcast},
    many reported attack strings lack practical effectiveness because they ignore data-flow constraints in the program, preventing them from reaching and triggering the vulnerable regex in actual applications.
Although traditional concolic execution can, in principle, solve such constraints and facilitate real-world exploitation,
    its application to ReDoS remains limited due to the intrinsic challenges of symbolically modeling complex regex semantics.
Thus,
    beyond early detection, 
    it is equally important to generate efficient and practical effective attack strings
    to bridge the gap between theoretical detection and exploitation.

To this end, 
    we present \sysname,
    a hybrid ReDoS attack string generation framework that provably constructs inputs inducing substantially greater matching cost than prior approaches while preserving practical exploitability.
Specifically, 
    we first observe that selecting appropriate characters and arranging them in a specific order within an attack string can force excessive backtracking and sharply increase matching cost;
    we then carry out a formal analysis and prove the conditions under which this effect arises and the specific forms of character selection and ordering that induce it.
Guided by the analysis,
    we derive three regex vulnerable patterns, 
    design corresponding attack string generation rules, 
    and implemented them in \sysname to construct candidate attack strings for each vulnerable regex accordingly.
\sysname further refines each candidate via a ReDoS-specific compositional concolic execution to ensure exploitability under realistic constraints.
In implementation, \sysname can be integrated into any ReDoS detector as a plugin attack generator with moderate additional overhead.

In our evaluation,
    we compare \sysname against a state-of-the-art ReDoS attack string generator \tool{RENGAR}~\cite{wang2023effective} on both 17{,}962 vulnerable regexes and 31 exploited ReDoS CVEs. 
The results demonstrate clear performance advantages for \sysname. 
Empowered by rigorous reasoning and compositional concolic refinement,
     \sysname generates attack strings that are on average $97.2\times$ to $3{,}872.4\times$ shorter than those produced by \tool{RENGAR} while achieving the same matching time thresholds of $0.1\mathrm{s}$, $1\mathrm{s}$, and $10\mathrm{s}$.
     \sysname also successfully reproduced 96.8\% of the exploited ReDoS CVEs while \tool{RENGAR} can only exploit 77.4\%.
To further assess \sysname's capability to exploit previously undisclosed vulnerabilities,
    we apply \sysname to 12 open source projects with more than 10k monthly downloads and discovered 59 exploitable ReDoS vulnerabilities across them, demonstrating \sysname’s practical effectiveness in real-world software systems.

In summary, we make the following contributions:

\begin{itemize}[leftmargin=*]
    \item We formally prove an observed effect that selecting specific characters and ordering them in attack strings amplifies regex matching cost; based on this,
    we derive three regex vulnerable pattern and devise provably efficient attack string generation rules.

    \item We develop \sysname, a hybrid ReDoS attack string generator that encodes formal generation rules into a scalable pipeline: it identifies vulnerable patterns, synthesizes compact attack strings, and uses compositional concolic refinement and ReDoS-specific function summaries to ensure real-world exploitability.
    
    \item 
    We evaluate \sysname on 17,962 vulnerable regexes, 31 exploited ReDoS CVEs, and 12 open-source projects. \sysname produces attack strings 97.2×–3,872.4× more efficient than the baseline tool, reproduces 96.8\% of exploited CVEs, and uncovers 59 previously undisclosed vulnerabilities across 12 projects; by submission, developers from six projects have validated 21 of them.
\end{itemize}

%% file: Tex/background.tex
\subsection{Regular Expression Denial of Service (ReDoS)}
A ReDoS attack forces a regex matcher to explore a polynomial- or exponential-scale number of alternative matches, consuming CPU resources and rendering the service unresponsive.
It is triggered by supplying an input that drives the matcher into its worst-case behavior, causing repeated backtracking over all possible match paths.
To execute such an attack, one must analyze the structure of regexes, identify vulnerable patterns, and generate attack strings accordingly.
A formal understanding of regexes is therefore required.

\noindent
\textbf{Regular Expression (Regex).}
In formal language theory~\cite{wang2010modular,hopcroft2006introduction}, 
    a regex is a formalism for describing string patterns and can be converted into an equivalent automaton~\cite{kleene1956representation}; 
    it explores all possible state paths in the automaton to check whether a given string matches a specified pattern.
Let $\Sigma$ denote a finite alphabet, and let $\Sigma^*$ be the set of all finite strings over $\Sigma$.
The symbol $\varepsilon$ represents the empty string and $\emptyset$ the empty set. 
We use $\mathbb{N}$ to denote the set of natural numbers and $\infty$ to represent infinity. 
A regex over $\Sigma$ is generated by the grammar:
\[
r ::= a \mid [C] \mid \varepsilon \mid \emptyset \mid (r_1'r_2') \mid (r_1' \mid r_2') \mid r'^{Q}.
\]

\noindent
where \(a \in \Sigma\), \(C \subseteq \Sigma\) is a character set, juxtaposition denotes concatenation and “\(|\)” denotes alternation.
$r'$ is a \emph{subregex} of $r$, i.e., a syntactic fragment of 
$r$ corresponding to a connected subgraph of its automaton.
The language $\mathcal{L}(r) \subseteq \Sigma^*$ of a regex $r$ 
is defined as the set of all strings over $\Sigma$ that are accepted by $r$.

In practice,
    regexes are usually used to match inputs of varying lengths. 
To achieve this, 
    a regex often contains subregexes  of the form \(r'^{Q}\), 
where \(Q \in (*, ?, +, \{m,\}, \{m,n\})\)  defines how characters matched by \(r'\) repeat within the strings of \(\mathcal{L}(r)\):
“?” denotes optional occurrence, “*”, “+” denote zero or more and one or more repetitions, respectively, and $\{m,\}$ and  $\{m,n\}$ where \(m,n \in \mathbb{N} \cup \{\infty\}\)  denote bounded repetitions.
We define such quantified subregexes as \emph{Loop Expressions}.

\begin{definition}[\textnormal{\textbf{Loop Expression (\func{LE})}}]
\label{def:le}
Let \(r'\) be a subregex of \(r\).
A \emph{Loop Expression} (\func{LE}) refers to a subregex \(r'\) that is annotated with a quantifier symbol or a range such as
\(*, ?, +, \{m,\}, \text{or } \{m,n\}\).
Each \func{LE} is associated with an alphabet \(\Sigma(\func{LE})\),  
    defined as the set of input symbols it can match.
For example, for \(r = a^*b\), the subregex \(r' = a^*\) is a \func{LE} and $\Sigma(\func{LE}) = ['a']$.
\end{definition}

The unfolding function $\mathcal{U}$ takes a regex $r$ as input,
    recursively splits $r$ by ``$|$'' until no further nested alternation remains, 
    and returns a set of subregexes of $r$ obtained through the splits.
Formally, $\mathcal{U}(r)=\{r_1',\dots,r_k'\}$.
For a \func{LE}, 
the unfolding conditionally removes or rewrites its quantifier before splitting its body by “\(|\)”, 
producing \(\mathcal{U}(LE)=\{r_1',\dots,r_l'\}\).

However, 
    as \func{LE}s can match variable length substrings, 
    the regex matching engine must dynamically explore all potential consumption boundaries for each \func{LE} in the input string, 
    forming the core of regex matching.

\noindent
\textbf{Regex Matching.}
Formally, 
    a regex matching function $\mathcal{M}(r, s)$ takes a regex $r$ and a concrete string $s \in \Sigma^*$ as input and returns a Boolean indicating whether $s \in \mathcal{L}(r)$.
Spencer’s~\cite{spencer1994regular} (i.e., backtracking) and Thompson’s~\cite{thompson1968programming} algorithms represent two classical approaches to regex matching. 
The former performs a constrained depth-first search,
    whereas the latter employs a breadth-first traversal.
Among them, 
    Spencer’s algorithm is widely adopted due to its implementation simplicity~\cite{davis2019rethinking}.
    
Specifically, 
    Spencer’s algorithm greedily extends each \func{LE} to consume its longest possible match,
    where we define the termination position as a stop point. 
Upon a matching failure, 
    it backtracks each \func{LE} with an untried stop point to resume matching.

\begin{definition}[\textnormal{\textbf{Stop Point}}]
\label{def:stoppoint}
A \emph{stop point} \(p\) is an index in the input string \(s\)
that marks the end of a valid match for the current $LE$,
formally satisfying 
\(LE_i \text{ matches } s[p_{i-1}:p_i]\) and \(\neg(LE \text{ matches } s[p_{i}\!+\!1])\) when \(s \in \mathcal{L}(r)\).
For example, for \(LE = a^*\) in \(r = a^*b\) and input \(s = aaa..ab\),
the stop point is at index -2.
\end{definition}

Formally, 
    given a regex $r$ and input $s\in\Sigma^*$, 
    Spencer’s algorithm processes $r$ as an ordered sequence of subregexes $(r_1',\dots,r_k')$. 
For each $r_i'$, 
    the matcher scans $s$ left to right,
    consuming the longest prefix matching  $r_i'$. 
If $r_i'$ is a \func{LE}, 
    it greedily locates its last stop point $p_i$ on $s$ and
    then recursively matches $s[p_i\!+1:]$ against $r_{i+1}'\dots r_k'$. 
If the overall match fails for $s$,
    the matcher backtracks the rightmost \func{LE} with an untried stop point \(p_i' < p_i\),
    sets \(p_i'\) as its new stop point,
    and resumes matching from $p_i'$.

\noindent
\textbf{ReDoS.}
While easy to implement,
    Spencer’s algorithm is vulnerable to \emph{catastrophic backtracking}, 
which occurs when matcher performs polynomial or exponential backtracking before concluding $s\notin\mathcal{L}(r)$.  
For example,
    when there are $k$ adjacent \func{LE}s in $r$ and all of them can consume the same substring of $s$ with length $n$, 
    the matcher must try every possible split of these $n$ characters among the $k$ \func{LE}s~\cite{li2022regexscalpel}, 
    yielding $O(n^{k-1})$ backtracking steps (\emph{polynomial backtracking}). 
If the \func{LE}s are nested, 
    each backtracking of an outer loop may trigger a complete re-exploration of its inner loop, 
    doubling the search space at each level 
    and resulting in exponential backtracking blowup, 
    for example $O(2^n)$.

\begin{definition}[\textnormal{\textbf{ReDoS}}]
\label{def:redos}
A \emph{ReDoS} is caused by adjacent $LE$s consume overlapping portions of the same input substring, 
causing the matcher to explore an exponential or polynomial number of backtracking paths 
before rejecting.
For example, 
    for \(r = .^*a.^*b\) and \(s_a = aa...ac\), 
    the matcher tries different splits of ``\(aa..a\)'' between the two $LE$s before failing.
\end{definition}
\noindent
In real-world systems,
    when vulnerable regexes are used to process untrusted inputs, 
    malicious attack strings can cause the system to hang by inducing excessive backtracking, 
    leading to Denial-of-Service attacks, or ReDoS attacks.

\subsection{ReDoS Attack String Generation}
\label{sec:bg_redos}
Given the prevalence of ReDoS vulnerabilities, many studies have proposed methods to detect such vulnerabilities and to generate corresponding attack strings for open-source software.
At a high level, for each vulnerable regex \(r\), an attack string \(s_a\) of length \(|s_a|\) is generated to trigger catastrophic backtracking and exploit the denial-of-service behavior. 
To generate a valid \(s_a\), prior work decomposes \(r\) into \(r = \varphi_1\ \varphi_2\ \varphi_3\), where \(\varphi_1\) is the \emph{prefix} occurring at the start of \(r\), \(\varphi_2\) is the \emph{infix} immediately following \(\varphi_1\) and containing all \func{LE}s that induce backtracking, and \(\varphi_3\) is the \emph{suffix} occurring after \(\varphi_2\).

\begin{definition}[\textnormal{\textbf{Repetition Unit}}]
\label{def:repunit}
A \emph{repetition unit} \(y\) for a attack string is a nonempty string \(y\in\mathcal{L}(\varphi_2)\) such that \(y^q\) is well defined for some \(q\ge 1\) and whose excessive repetition causes the \func{LE}s in \(\varphi_2\) to exhibit exhaustive backtracking.
For example, 
    for \(r = .^*a.^*b\) with \(s_a = aaa..ac\),
    the repetition unit is \(y = a\).
\end{definition}

The attack string is constructed as \(s_a = x\,y^{q}\,z\), such that \(x \in \mathcal{L}(\varphi_1)\), repetition unit \(y\in\mathcal{L}(\varphi_2)\) which will be repeated $q$ times, and \(z \notin \mathcal{L}(\varphi_3)\), thus \(s_a \notin \mathcal{L}(r)\), forcing exhaustive backtracking~\cite{wang2023effective}. 
Specifically, existing attack string generation methods fall into three categories:

\noindent
\textbf{Static Analysis.}
Static analysis based generators such as \tool{regexploit} detect syntactic patterns of vulnerable regexes, for example adjacent or nested \func{LE}s, and apply predefined construction rules to rapidly produce candidate attack strings~\cite{doyensec_regexploit_2021}. 
Automata-based approaches instead translate the regex into an automaton and reason about matching behaviour, offering stronger formal guarantees~\cite{kirrage2013static,parolini2023sound,rathnayake2014static}. 
While these methods are computationally cheap and produce candidates quickly, their precision is limited because they do not capture the full runtime semantics of regex matching.

\noindent
\textbf{Dynamic Exploration.}
Dynamic exploration based methods treat attack generation as a search problem, executing the regex over many mutated inputs and observing execution matching costs to identify inputs that induce high matching cost~\cite{shen2018rescue,barlas2022exploiting,mclaughlin2022regulator,petsios2017slowfuzz}.
However, they suffer from two major limitations: 
(1) they require substantial computational resources to evaluate a large number of candidates, and 
(2) their input space coverage remains limited, which may cause them to miss inputs that trigger higher regex matching costs.

\noindent
\textbf{Synthesis with Executional Validation (Hybrid Methods).}
To balance coverage and computational overhead, recent tools adopt executional validation mechanisms. 
At the regex level, hybrid approaches combine pattern-guided candidate generation with dynamic execution validation~\cite{davis2018impact,wang2023effective,noller2018badger}. 
Candidate attack strings are first generated statically and then executed dynamically to retain only those whose matching time exceeds a predefined threshold. 
Although effective for trigger vulnerable regexes in isolation, these methods often fail to capture the propagation and data-flow constraints that determine whether an input can reach the vulnerable regex in real-world programs.

In traditional vulnerability exploitation,
    at the program level,
    concolic execution enables such real-world effectiveness validation by simultaneously tracking concrete inputs and symbolic expressions, 
    and solving path constraints to explore new execution paths and trigger target behaviors (e.g. vulnerability manifest)~\cite{10.1145/1065010.1065036,10.5555/1792734.1792766,barbosa2022cvc5}. 
While classical DART-style~\cite{10.1145/1065010.1065036,cadar2008klee} concolic executors face exponential path explosion when analyzing programs with multiple functions and branching paths,
SMART-style~\cite{godefroid2007compositional} compositional concolic executors mitigates this by testing functions in isolation, summarizing their behaviors, and reusing these summaries in higher-level analyses~\cite{kim2019target}. 
This strategy reduces redundant exploration and improves scalability; however,
    the inherent difficulty of symbolically modeling regex semantics can still cause the number of summary entries to grow exponentially, reintroducing scalability challenges when exploiting ReDoS in large programs.
Consequently, concolic execution is seldom used for ReDoS attack string generation, despite wide use in other vulnerability classes.

\subsection{Motivation}

Existing works typically generate attack strings by analyzing the regexes in isolation,
    producing inputs that merely exceed the time consumption threshold.
This approach raises two concerns: 
    (i) whether attack strings of similar length that\textbf{ incur higher matching cost} exist, and (ii) whether generated inputs can \textbf{exploit ReDoS in real-world systems}.
In other words, 
   they focus only on the success of attack string generation but ignore both the efficiency and the practical effectiveness.
Previous studies and discussions have shown that, 
    without an efficient and effective attack string, 
    developers will often assign the reported vulnerability low priority ~\cite{bhuiyan2025sok}, mark it as low severity~\cite{CVE-2017-16137,CVE-2017-16098}, 
    or reject it~\cite{CVE-2023-39663} as they consider the issue unlikely to be exploitable in real systems,
    leaving the underlying risk unresolved.
For instance, 
    some existing tools require inputs of 100K to 1M characters to trigger a 10-second slowdown, yet common deployment environments enforce strict input constraints such as the 8 KB HTTP header limit in nginx and Apache Tomcat~\cite{bhuiyan2025sok}, rendering such attacks impractical in practice.
    
\begin{figure}[htbp]
    \centering
    \includegraphics[width=\linewidth]{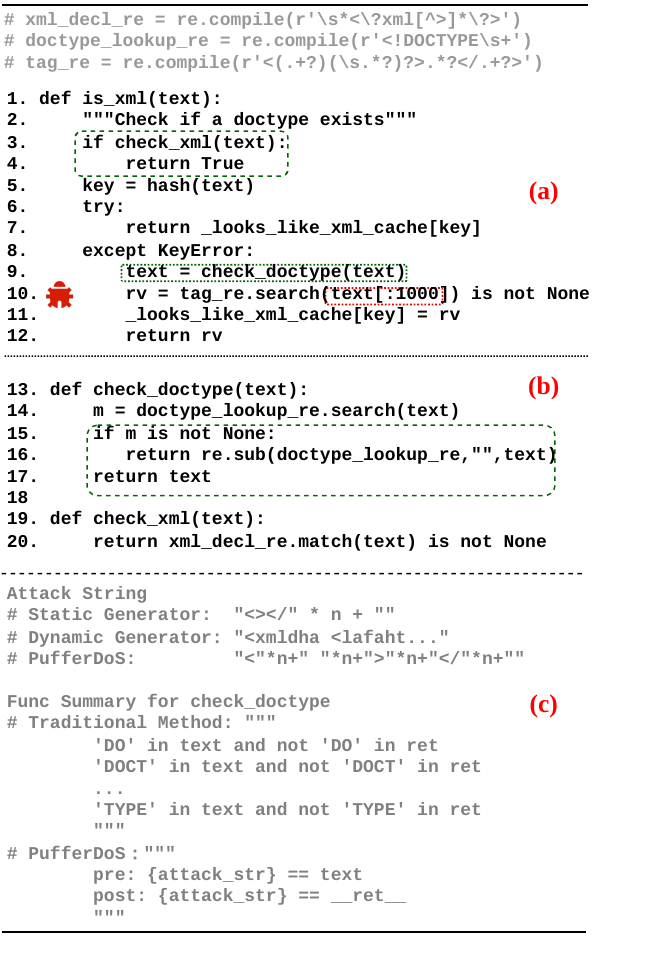}
    \caption{A motivating example.}
    \label{fig:motivation}
\end{figure}

Figure~\ref{fig:motivation} (a) and (b) shows a representative ReDoS vulnerability that we successfully exploited but was missed by existing approaches.
It originates from a syntax-highlighting library with monthly downloads exceeding 300 million.
To preserve anonymity,
    we obfuscated the function name and masked certain code fragments without altering the logical structure of the vulnerability. 
The vulnerable regex \func{tag\_re} is used at Line 10.
Existing tools leverages static analysis-guided attack string generation will decompose \func{tag\_re} into $\varphi_1 = \emptyset,\,
\varphi_2 = \text{\texttt{<(.+?)(\textbackslash s.*?)?>.*?</.+?}}$ and $\varphi_3 = \text{\texttt{>}}$. 
As shown in Figure~\ref{fig:motivation} (c),
    an attack string \(s_a=x\,y^{n}\,z\) is constructed with \(x = \varepsilon\), repetition unit \(y = \texttt{<><\textbackslash}\) satisfying \(y \in \mathcal{L}(\varphi_2)\) and \(z = \varepsilon\) such that \(z \notin \mathcal{L}(\varphi_3)\). 
Hence \(s_a \notin \mathcal{L}(\func{tag\_re})\) and excessive backtracking is triggered.
Dynamic exploration-based generators iteratively mutate and validate candidate strings against \func{tag\_re} until a substantial slowdown in matching is observed.
However, 
    our manual testing revealed that these generators face practical limitations. 
While the static analysis-guided generated attack string \(s_a\) with \(|s_a|=500,000\) induces approximately a 10 seconds hang when the regex is evaluated in isolation,
    the program truncates inputs to 1,000 characters at Line 10,
    reducing the slowdown to only 0.1s in practice. 
Dynamic exploration-based generators encounter similar obstacles:   
    truncation prevents exploits from manifesting, 
    and earlier filters at Lines 3 and Line 9 may block randomly generated inputs before they reach Line 10.
According to prior study~\cite{bhuiyan2025sok,davis2018impact} and developer feedback,
    attacks causing slowdowns below 10 seconds, involving excessively long inputs, 
    or failing to reach the vulnerable regex are deemed ineligible and therefore ignored.

Traditional concolic execution methods also fail to handle this example as it involves complex regexes.
Although SMT solvers like Z3 provide regex support,
    they perform poorly when reasoning about multiple complex regex constraints simultaneously~\cite{stanford2021symbolic}, which exacerbates the path explosion problem in DART.
Compositional methods are likewise inadequate, 
    as summarizing functions that embed regexes requires symbolically modeling all possible regex behaviors, 
    causing substantial overhead when reusing function summaries.
As shown in Figure~\ref{fig:motivation} (c), 
    for function \func{check\_doctype} in Figure~\ref{fig:motivation} (b),
    attempting to capture all regex behaviors results in large and intractable summaries.
    
These limitations motivate us to design a novel attack string generation approach that improves both efficiency and effectiveness. 
For efficiency,
    we seek \(s_a\) that incur higher matching time under a fixed length budget. 
Consider the motivating example,
    given \(\func{tag\_re}\) with many adjacent \( \func{LE} \)s, 
    a stop point \(p_i\) is selected on $s_a$ for \( \func{LE}_i \);
    the subsequent \( \func{LE}_{i+1} \) advances from \(p_i\) to locate its own stop point \(p_{i+1}\).
When matching fails, 
    the engine backtracks by decreasing \(p_{i+1}\) through the stop points in interval \([p_i,\, p_{i+1}]\). 
Therefore, 
    to improve efficiency,
    we seek an \(s_a\) that enlarges the initial gaps between \(p_i\) and \(p_{i+1}\), 
    forcing each backtrack to reprocess longer input segments,
    thereby increasing overall matching cost.
For effectiveness, once a candidate \(s_a\) is obtained, 
    concolic execution is used to refine and validate \(s_a\) so that it reaches the \(r\) in realistic program contexts. 
To scale concolic execution, 
    ReDoS-specific function summaries are synthesized to retain only summary entries that affect propagation of \(s_a\), 
    discarding unrelated behavior.

As shown in Figure~\ref{fig:motivation} (c),
    a stronger $s_a$ can be constructed by repeating each character in repetition unit \(y = \texttt{<><\textbackslash}\) in isolation then concatenating the results to enlarge gaps between stop points,
    and finally adding a suitable prefix and a disruptive suffix.
The function summary for \func{check\_doctype} is refined to the contract:
    if the parameter is the attack string, then the return value must also equals to it. 
The abstraction eliminates in-function regex reasoning but preserves propagation semantics.
We then use concolic execution on the function summary to refine $s_a$ so it bypasses the filters at Lines 3 and 9.

During testing, 
    the refined $s_a$ produced a system hang exceeding 4,000 seconds when supplied as input with \(|s_a|=1,000\),
    demonstrating over 40,000× efficiency gain compared to traditional methods and confirming the vulnerability’s real-world exploitability.

%% file: Tex/study.tex
To formalize the insights from the motivating example, we present a formal analysis that yields a heuristic model of matching cost to guide our subsequent tool design. Note that, \textbf{complete proofs} of all lemmas and the theorem are presented in the Appendix~\ref{app:proof}. 


Figure~\ref{fig:example} presents a vulnerable regex \(r=\texttt{a+.*(cd|e)+f}\) and five corresponding representative attack strings \(s_{a0}, s_{a1}, s_{a2}, s_{a3}, s_{a4}\),
    which are used throughout this section to illustrate each step of our verification.
Among these attack strings, $s_{a0}$ repeats its repetition unit only once while other strings repeat them multiple times.
For the repetition units,
    \(\func{LE}_0\) consistently selects the repetition unit \(y_0 = a\). 
    \(\func{LE}_1\) coincides with \(\func{LE}_0\) except in \(s_{a2}\) where it instead selects \(y_1 = b\).
    \(\func{LE}_2\) chooses \(y_2 = cd\) in \(s_{a1}\) and \(s_{a2}\), 
    while choosing \(y_2 = e\) in the remaining strings.
Regarding structural layout,
    \(s_{a0},s_{a1}, s_{a2}, s_{a3}\) arranges repetition units in an interleaved manner, 
    whereas \(s_{a4}\) concatenates them separately.

\noindent
\begin{assumption}\label{ass:1}[Monotonic Matching Cost in Reachable Stop Points]
Fix a regex $r$, a backtracking matcher, the set of participating \func{LE}s inside $\varphi_2$, and the failing suffix location (i.e., where $\varphi_3$ first rejects). Let $C_r(s_a,y,n,m)$ denote the matching cost for input $s_a$, where $n=|s_a|$ and $m$ is the total number of reachable stop points contributed by the participating \func{LE}s on $s_a$ to the left of the failing location. 
Here $y$ is any repetition unit used only to describe the layout of reachable stop points inside $\varphi_2$.
If two inputs $s_a$ and $s_a'$ keep the participating \func{LE}s and the failing location unchanged and $m' \ge m$, then
\[
C_r(s_a',y',n',m') \;\ge\; C_r(s_a,y,n,m),
\]
with a strict inequality whenever $m' > m$.

\noindent
Note that we assume these $LE$s are unbounded (e.g. with quantifier *,?,+) or possess a high repetition capacity (e.g.,  with quantifier \{m, n\} where $n - m$ is sufficiently large), 
    such that they are capable of consuming arbitrarily long substrings within $\varphi_2$.
\end{assumption}

To demonstrate this heuristic, 
we show how the number of stop points affects the matching cost. 
In Figure~\ref{fig:example}, 
    for \(\func{LE}_0\), there are only \textbf{one} stop point in \(s_{a0}\) and \textbf{two} in \(s_{a1}\). 
The regex matcher first attempts to match using the last stop point of \(\func{LE}_0\); 
when this attempt fails, the matcher backtracks once to try the previous stop point in \(s_{a1}\), 
whereas in \(s_{a0}\) the matching terminates immediately, 
Thus, \(s_{a1}\) incurs a higher matching cost than \(s_{a0}\).

\noindent
\begin{lemma}[Shorter Repetition Units Increase Matching Cost]\label{lem:1}
Fix a regex $r$ and its infix $\varphi_2$ that induces backtracking. 
Let $x\in\mathcal{L}(\varphi_1)$ and $z\notin\mathcal{L}(\varphi_3)$ be fixed so that the resulting attack strings fail at the same suffix location. 
If $y_1,y_2\in\mathcal{L}(\varphi_2)$ are feasible repetition units for all participating \func{LE}s and $|y_1|<|y_2|$, 
    then for any fixed length budget $n$, constructing the infix with $y_1$ yields no lower matching cost than with $y_2$, and the cost is strictly higher whenever the number of reachable stop points increases.
\end{lemma}

In Figure~\ref{fig:example}, 
\(s_{a2}\) uses \(y = acd\) to construct the infix, whereas \(s_{a3}\) uses \(y = ae\). 
As a result, 
    for \(\func{LE}_0\),
    there are \textbf{three} stop points in \(s_{a2}\) and \textbf{four} in \(s_{a3}\).
By Heuristic~\ref{ass:1} (matching cost increases with the number of reachable stop points), 
\(s_{a3}\) incurs a higher matching cost under a fixed attack string length.

\noindent
\begin{lemma}[Repetition Units from Intersection Increase Matching Cost]\label{lem:2}
Let $r$ contain adjacent $\func{LE}_i$ and $\func{LE}_{i+1}$ with $\Sigma(\func{LE}_i)\cap\Sigma(\func{LE}_{i+1})\neq\varnothing$. 
Fix $x\in\mathcal{L}(\varphi_1)$ and $z\notin\mathcal{L}(\varphi_3)$. 
If $y_1,y_2\in\mathcal{L}(\varphi_2)$ are feasible repetition units such that every symbol of $y_1$ lies in $\Sigma(\func{LE}_i)\cap\Sigma(\func{LE}_{i+1})$ (both loops can consume $y_1$), while $y_2$ is consumable by exactly one of the two loops, i.e.,
\[
\;y_2\in\Sigma(\func{LE}_i)\setminus\Sigma(\func{LE}_{i+1})\;\;\text{ or }\;\; y_2\in\Sigma(\func{LE}_{i+1})\setminus\Sigma(\func{LE}_i),
\]
then 
    for any fixed length budget $n$, constructing the infix with $y_1$ yields no lower matching cost than with $y_2$, 
    and the cost is strictly higher whenever the number of reachable stop points increases.
\end{lemma}

\begin{figure}[htbp]
    \centering
    \includegraphics[width=0.8\linewidth]{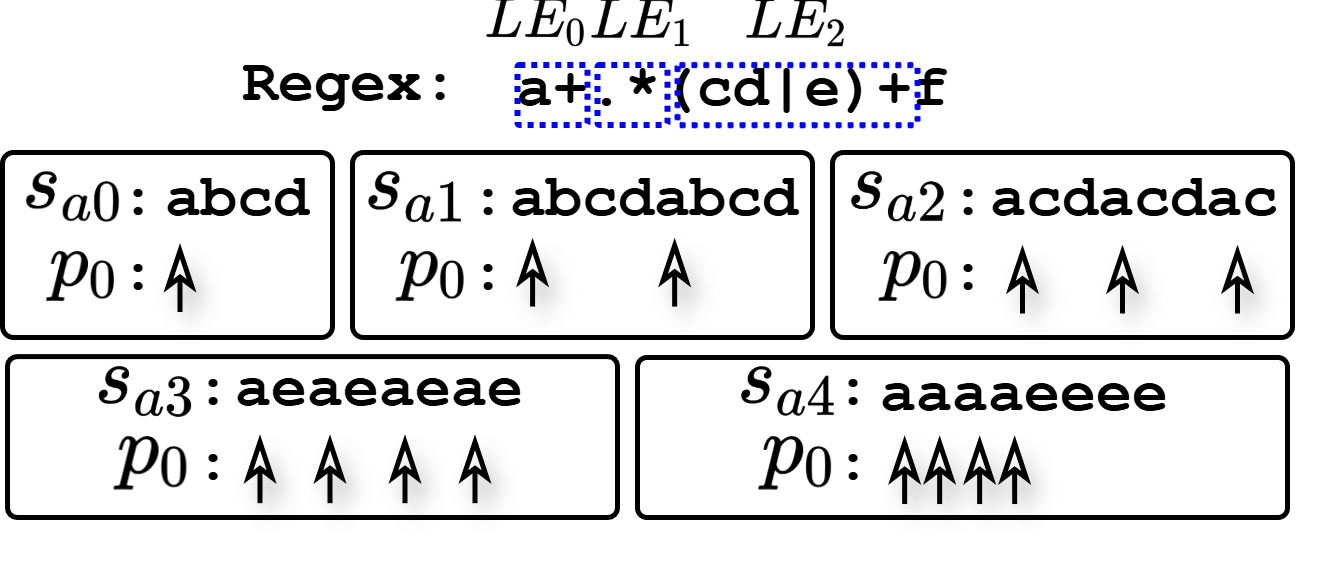}
    \caption{Illustrative example of attack string generation.}
    \label{fig:example}
\end{figure}

To exemplify Lemma~\ref{lem:2},
    we consider $s_{a1}$ and $s_{a2}$ in Figure~\ref{fig:example}.
    If \(LE_0\) and \(LE_1\) select different repetition units, $a$ and $b$ respectively under intersecting alphabets condition as in $s_{a1}$, rather than sharing the same unit $a$ as in $s_{a2}$,
    the stop points of \(LE_0\) reduce from \textbf{three} in \(s_{a2}\) to \textbf{two} in \(s_{a1}\).
By Heuristic~\ref{ass:1}, \(s_{a2}\) incurs a higher matching cost.

\noindent
\begin{lemma}[Separated Stop Points Increase Matching Cost]\label{lem:3}
Let $r$ contain adjacent $\func{LE}_i$ and $\func{LE}_{i+1}$ with $\Sigma(\func{LE}_i)\cap\Sigma(\func{LE}_{i+1})\neq\varnothing$,
and they can both consume the same substring of $s_a$ as in Definition~\ref{def:redos}.
Fix $x\in\mathcal{L}(\varphi_1)$ and $z\notin\mathcal{L}(\varphi_3)$. 
Denote by $\mathcal{P}_i$ and $\mathcal{P}_{i+1}$ the sets of stop points of $\func{LE}_i$ and $\func{LE}_{i+1}$ within the infix, respectively.
When constructing infix,
    if the stop points in $\mathcal{P}_i$ and $\mathcal{P}_{i+1}$ are \emph{separated} (i.e., all stop points of $\func{LE}_i$ occur before those of $\func{LE}_{i+1}$) 
    rather than \emph{intervaled} (i.e., the stop points of the two \func{LE}s interleave along the string),
    then for any fixed length budget $n$, 
    it yields no lower matching cost.
\end{lemma}

For Lemma~\ref{lem:3}, 
we refer to \(s_{a3}\) and \(s_{a4}\) in Figure~\ref{fig:example}. 
When the regex matcher evaluates \(\func{LE}_0\) using the last stop point in each string (\(s_{a3}[6]\), \(s_{a4}[3]\)), 
a failed match causes the matcher to backtrack \textbf{once} at \(s_{a3}[7]\) but \textbf{four} times at \(s_{a4}[4:7]\) to select the stop points of \(\func{LE}_1\).
The same trend holds when \(\func{LE}_0\) chooses other stop points,
    resulting in an increased matching cost in $s_{a4}$.

\noindent
\begin{theorem}[Principles for Cost-Increasing Attack String Generation]\label{thm:1}
An attack string attains higher matching cost when each $\func{LE}$ employs the shortest feasible repetition unit that is shared by as many $\func{LE}$s as possible.  
If two adjacent \func{LE}s have intersecting alphabets, 
    a repetition unit should be selected from their alphabet intersection.   
Each selected repetition unit should be repeated independently and then concatenated to ensure the stop points are separated.
\end{theorem}
\vspace{-1em}

%% file: Tex/pufferdos.tex
In this section, we present \sysname, a hybrid framework that operationalizes the formal foundations of attack generation into an efficient, structure-driven attack string generation pipeline (Section~\ref{sec:att_str_gen}) and augments it with ReDoS-specific, compositional concolic refinement to ensure practical effectiveness (Section~\ref{sec:concolic}).
Figure~\ref{fig:workflow} presents an overview of \sysname, which consists of three core components.

\noindent\textbf{(1) Vulnerable Regex Detection:} 
Given a target program, 
    we extract all regexes from the source code by locating regex-related API calls. 
We then apply ReDoS detectors to identify vulnerable regexes. 

\noindent\textbf{(2) Structure-Driven Attack String Generation:} 
Building upon the formal foundations in Section~\ref{sec:study}, 
    we devise a taxonomy of three regex vulnerable patterns and instantiate pattern-specific, structure-driven rules for attack string generation. 
For each vulnerable regex $r$, 
    we parse $r$ into subregexes, identify the looping expressions (\func{LE}s), and compute their alphabets \(\Sigma(\func{LE})\). 
We then classify $r$ into the appropriate pattern and synthesize the infix under the corresponding rule. 
Finally, we derive the prefix and suffix from the non-loop segments (adjusting the suffix to ensure the string remains outside \(\mathcal{L}(r)\)) and concatenate the prefix, infix, and suffix to obtain a complete attack string.

\noindent\textbf{(3) ReDoS-Specific Concolic Refinement:}
To ensure the generated attack string can propagate through real program and reach the vulnerable regex invocation, 
    we first identify the target program entry points and construct call graphs to extract feasible call paths from entry point to the target regex. 
We then perform concolic unit testing to generate ReDoS-specific function summaries,   
    followed by a system-level concolic execution to resolve interprocedural constraints. 
The attack string is refined accordingly, 
    ensuring exploitability in realistic execution environments.


\begin{figure*}[t]
    \centering
    \includegraphics[width=0.97\textwidth]{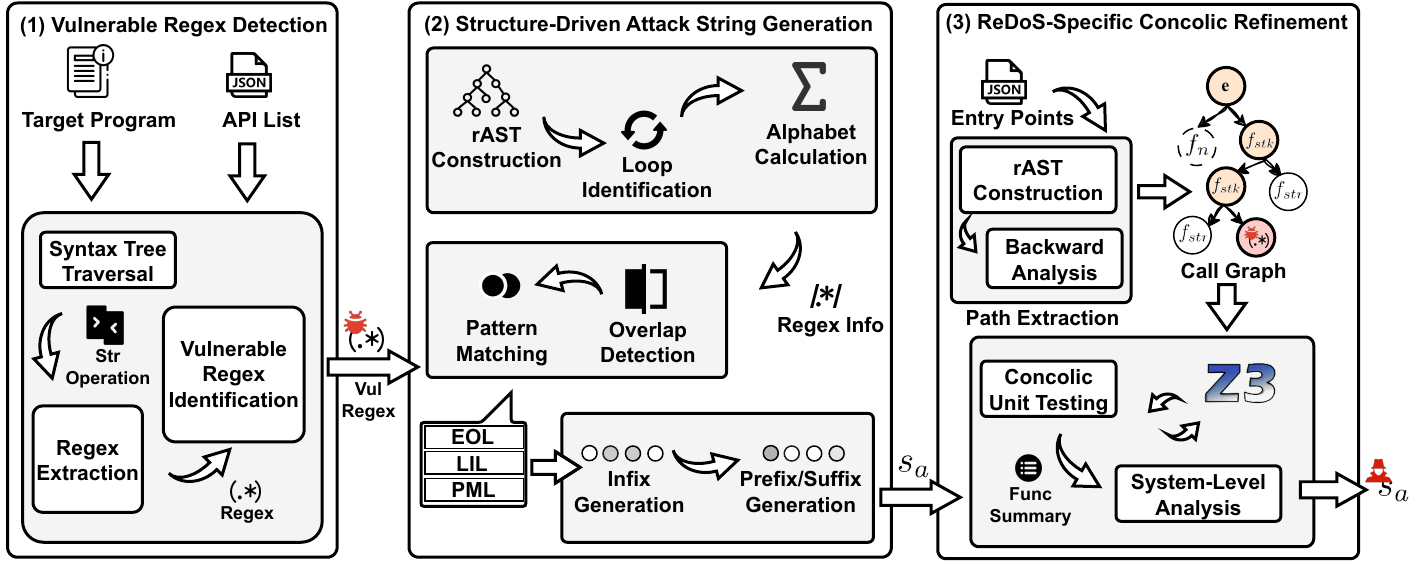}
    \caption{\sysname workflow.}
    \label{fig:workflow}
\end{figure*}

\vspace{-1em}
\subsection{Vulnerable Regex Detection\label{sec:detect}}
\vspace{-1em}
The vulnerable regex detection component takes a target program as input, statically analyzes the source code to locate calls to regex-related APIs (e.g., matching, compilation, substitution) through function name matching, and extracts the regexes provided as string-literal arguments of the APIs.

Many regexes are not passed directly to APIs but built dynamically through variables assignment, concatenation, or interpolation. To handle these, \sysname scans source files, traverses abstract syntax trees, and tracks string provenance~\cite{doyensec_regexploit_2021}, recursively resolving variable definitions to reconstruct the concrete regex when a regex API is invoked.

For each extracted regex, \sysname applies ReDoS detectors to identify vulnerable cases and retains only those confirmed as vulnerable for subsequent analysis.

\subsection{Structure-Driven Attack String Generation\label{sec:att_str_gen}}

Given the vulnerable regexes identified in Section~\ref{sec:detect}, this section describes how \sysname generates a high-cost attack string for each regex. To systematically approach attack generation, we first establish a pattern taxonomy that categorizes vulnerable regex constructs by their exploitation characteristics (Section~\ref{sec:taxmony}). We then detail how this taxonomy guides our attack string generation procedure (Section~\ref{sec:attgen}).


\subsubsection{ReDoS Pattern Taxonomy\label{sec:taxmony}}
Building on the matching cost analysis in Section~\ref{sec:study}, we propose a taxonomy of three vulnerable patterns that (i) cover all patterns identified in prior works (see Appendix~\ref{apx:pattern}) and (ii) prescribes pattern-specific attack string generation strategies. These strategies are tailored to markedly increase matching cost in backtracking regex engines.


\begin{pattern}[\textnormal{\textbf{Exponential One-Loop (EOL)}}]
An \emph{EOL} consists of a single pathological \func{LE} whose unfolding \(\mathcal{U}(\func{LE})\) contains at least two distinct subregexes \(r'_i, r'_j\) (\(i \ne j\)) such that \(\Sigma(r'_i) \cap \Sigma(r'_j) \ne \varnothing\). For example, regex \texttt{\string^(ab|a|b)+\$} exhibits EOL pattern.
\label{def:eol}
\end{pattern}

To generate attack string for $r$ with EOLs, 
    \sysname computes the overlap sets \(\Sigma(r'_i) \cap \Sigma(r'_j) \ne \varnothing\) for each EOL.
By Lemma~\ref{lem:1},
    if only one EOL exists, 
    it selects and repeats one of the shortest elements from its overlap set to form the infix.
For multiple EOLs, 
     by Theorem~\ref{thm:1},
    \sysname first computes the intersection of their overlap sets.
If the intersection is non-empty, 
    the shortest element from the intersection is chosen as repetition unit $y$ and repeated to construct infix; 
otherwise, 
    \sysname repeats the shortest element from each individual overlap set and concatenate the results to form an infix with separated stop points.

\begin{pattern}
[\textnormal{\textbf{Loop-Intersect-Loop (LIL)}}]
An \emph{LIL} consists of two adjacent \func{LE}s,
    separated by a non-loop segment \(S_e\), 
    written as \(r = \func{LE}_1\, S_e\, \func{LE}_2\). It holds that \(\Sigma(\func{LE}_1) \cap \Sigma(\func{LE}_2) \ne \varnothing\) and \(\Sigma(S_e) \subseteq \Sigma(\func{LE}_1) \cap \Sigma(\func{LE}_2)\).
    Separator \(S_e\) may be empty.
    For example, regex $.^*a.^*$  exhibits LIL pattern.
\label{def:lil}
\end{pattern}

To generate attack string,
    for each LIL pattern,
     \sysname sets repetition unit \(y = S_e\) as it can be matched by both \func{LE}s and enables them to consume overlapping substrings without being obstructed by the non-loop segment. 
If the segment is empty,
    one of the shortest element from \(\Sigma(LE_1) \cap \Sigma(LE_2)\) is selected instead. 
For multiple LILs, e.g., $r=\func{LE}_1\, S_{e_1}\, \func{LE}_2\, S_{e_2}\, \func{LE}_3$, by Theorem~\ref{thm:1}, 
    the repetition unit $y_i$ of each LIL pattern is repeated to form a base string \(b_i = y_i^q \),
    and all resulting $b_i$ are concatenated to construct the infix.

\begin{pattern}[\textnormal{\textbf{Polynomial Multi-Loop (PML)}}]
A \emph{PML} consists of \(k\) \func{LE}s, such that \(\forall i \in [1, k-1],\;
\Sigma(\func{LE}_i) \cap \Sigma(\func{LE}_{i+1}) = \varnothing
\;\text{or}\;
\Sigma(S_e) \subsetneq \Sigma(\func{LE}_i) \cap \Sigma(\func{LE}_{i+1})\). 
That is, adjacent \func{LE}s have disjoint alphabets or the non-loop segment cannot be accepted by both of the \func{LE}s.
For example, regex $\backslash s^*a\backslash s^*$ exhibits the PML pattern.

\label{def:pml}
\end{pattern}

Here, 
    $S_e$ denotes the non-loop segment between each \func{LE} pair $(\func{LE}_i, \func{LE}_{i+1})$.
For regex with PMLs, 
     \sysname falls back to the traditional attack string generation.  
We consider the target regex as \(r = \func{LE}_1,\dots,\func{LE}_k\),
    by Lemma~\ref{lem:1},
    \sysname selects the shortest symbol \(a_i\in\Sigma(\func{LE}_i)\) for each \func{LE},
    forms the base string \(b=a_1a_2\cdots a_k\) as repetition unit,
    and repeats the base string to build the infix.

\subsubsection{Attack String Generation\label{sec:attgen}} 
Algorithm~\ref{alg:attgen} shows the whole pipeline of generating attack string.
It takes the identified vulnerable regexes as input.
For each regex, 
it expands the regex structure (Lines 4--10),
matches the vulnerable patterns, and generates the corresponding $s_a$ (Lines 11--25).   

Since a regex may contain multiple alternatives introduced by the disjunction operator “$|$”, leading to distinct matching behaviors,
we define a \textit{regex path} (rPath) as an ordered, non-branching sequence of subregexes \(r'\) representing one specific alternative of \(r\).  
Given \(r\), \sysname expands it into a \textit{regex Abstract Syntax Tree} (rAST), where \(r\) is the root and internal nodes denote intermediate unfoldings.
Branching in the rAST arises from (i) “$|$” operators, with each disjunctive alternative as an independent child node, and (ii) \func{LE}s allowing zero repetitions (e.g., \verb|.*| or \verb|a{0,1024}|), each producing two children, one excluding the \func{LE} to represent zero repetitions and the other including it to represent one or more repetitions.
Although splitting original regex $r$ into rPaths  may omit certain branches of the original regex,
    any attack string synthesized for a given rPath remains a valid attack string for $r$, 
    as the language of each rPath $\mathcal{L}{(rPath)}$ is a subset of $\mathcal{L}{(r)}$.
\vspace{-0.2em}
\begin{figure}[htbp]
    \centering
    \includegraphics[width=0.92\linewidth]{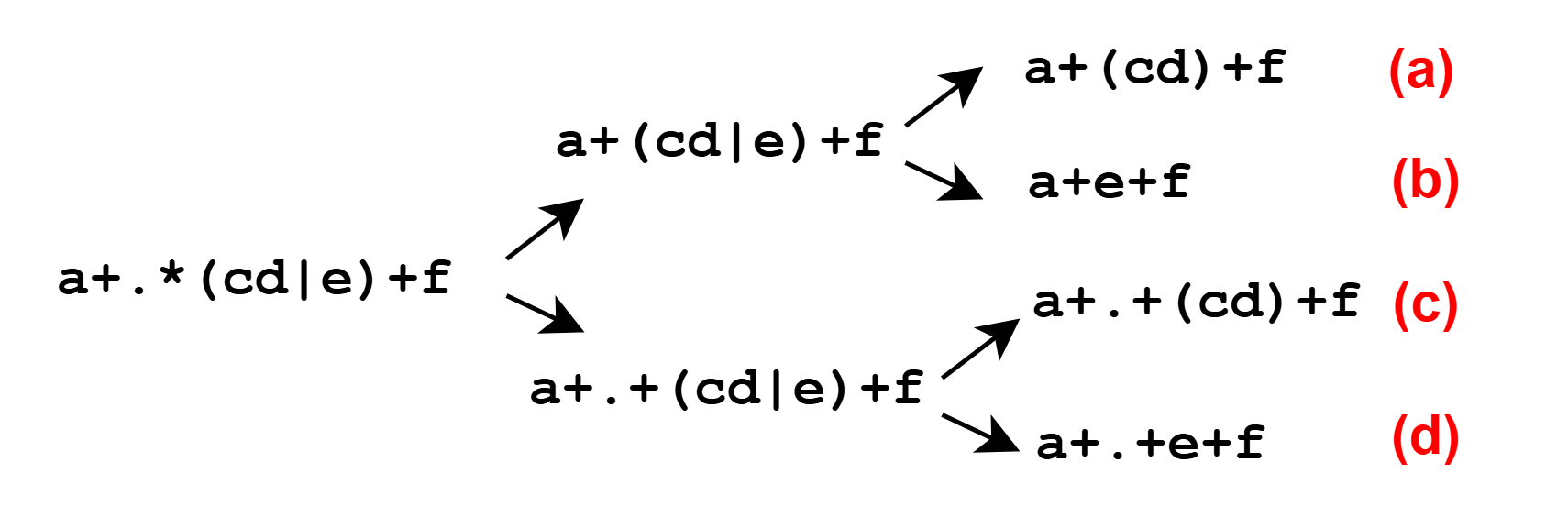}
    \caption{An illustrative example of rAST and rPath.}
    \label{fig:example2}
\end{figure}

During expansion, 
\sysname checks whether a \func{LE} can be unfolded into subregexes with overlapping alphabets that meet the EOL condition defined in Pattern~\ref{def:eol}.  
If found, it marks and reports the EOL pattern and terminates; otherwise, decomposition continues until all leaf nodes are produced, each corresponding to one rPath.
As an example, in Figure~\ref{fig:example2}, the regex is expanded into an rAST with four leaf nodes, each leaf node corresponding to one rPath.

For each rPath,
    \sysname sequentially traverses all subregexes in it.
Whenever two adjacent \func{LE} are encountered, 
    \sysname calculates \(\Sigma(\func{LE})\) and the overlap set \(\Sigma(\func{LE}_i) \cap \Sigma(\func{LE}_{i+1})\).
If the overlap set is non-empty,  
    \sysname examines the separating segment between the two \func{LE}s to confirm all characters within it belong to the overlap set.
A LIL pattern is identified whenever this criterion holds.  

When multiple rPaths contain EOL or LIL patterns, 
    \sysname selects the rPath with the largest number of such patterns and the shortest repetition unit; 
    if several satisfy these criteria equally, one is chosen arbitrarily. If no rPath contains an EOL or LIL pattern, the regex is classified as PML.
    For example, in Figure~\ref{fig:example2}, 
\sysname identifies that rPaths (a)--(d) all contain LIL patterns: 
(a) and (b) each contain one, while (c) and (d) contain two and are therefore shortlisted for comparison.

Once a vulnerable pattern is identified \sysname initiates $s_a$ construction by generating the infix, 
    applying pattern specific strategies for each case as described in Section~\ref{sec:taxmony}. After generating the infix, \sysname builds the prefix by collecting all characters occurring before the first \func{LE}. The suffix is formed from the characters after the last \func{LE} with the final character removed, ensuring $s_a\notin\mathcal{L}(r)$.

\begin{algorithm}[t]
  \caption{\textsc{Attack String Generation}}
  \label{alg:attgen}
  \DontPrintSemicolon
  \KwIn{$R$}
  \KwOut{$S$}

  $S \gets \emptyset$\;

  \ForEach{$r \in R$}{
    $EOL\_flag \gets \textbf{false}$; 
    $LIL\_flag \gets \textbf{false}$; 
    $det\_p \gets \varnothing$; 
    $rAST \gets r$\;

    \While{$\exists\, n \in rAST \land n.\textit{expandable} = \textbf{true}$}{
  \uIf{``$|$'' $\in n$}{
    $rAST \gets \textsc{getSplitChild}(n)$\;
  }
  \ElseIf{n is \func{LE} allowing zero reps}{
    $rAST \gets \textsc{getLEChild}(n)$\;
  }
  \If{$n\;  matches\;   Pattern\;   1$}{
    $EOL\_flag \gets \textbf{true}$; 
    $det\_p \gets n$; 
    \textbf{break}\;
  }
}

    \If{$EOL\_flag$}{
      $w \gets \textsc{InfixEOL}(det\_p)$\;
    }
    \Else{
      $P \gets \textsc{getRPATH}(rAST)$\;
      \ForEach{$p \in P$}{
        \uIf{$p\; matches\; Pattern\; 2$}{
          Tag($p$) $\gets$ LIL; 
          $LIL\_flag \gets \textbf{true}$\;
        }
      }
      \If{\textbf{not} $LIL\_flag$}{
        Tag($p$) $\gets$ PML\;
      }
      $\hat{p} \gets \textsc{SelectPath}(\{p \mid Tag(p) \neq \text{null}\})$\;
      \uIf{Tag($\hat{p}$) = LIL}{
        $w \gets \textsc{InfixLIL}(\hat{p})$\;
      }
      \Else{
        $w \gets \textsc{InfixPML}(\hat{p})$\;
      }
    }

    $pre \gets \textsc{Prefix}(w, r)$; 
    $suf \gets \textsc{Suffix}(w, r)$; 
    $s_a \gets pre \cdot w \cdot suf$\;
    $S \gets S \cup \{(r, s_a)\}$\;
  }

  \Return{$S$}
\end{algorithm}




\subsection{ReDoS-Specific Concolic Refinement\label{sec:concolic}}

Once an attack string $s_a$ is constructed,
    \sysname performs compositional concolic execution to systematically refine it toward inputs that trigger the vulnerable regexes.

As shown in Figure~\ref{fig:workflow}, \sysname takes as input the target program $P$, 
    a set of vulnerable regex $r$, and their corresponding $s_a$. 
It operates in three phases: 
    (1) extracting execution paths $p$ from program entry points $e$ to each $r$ and record all reachable functions $f$.
    (2) performing concolic unit testing on selected functions from $f$ along $p$ to derive function  summaries $\phi_f$, and 
    (3) composing system-level concolic execution by reusing $\phi_f$ to refine $s_a$ 

\noindent  
\textbf{(1) ReDoS-Reachable Path Extraction.}
\sysname constructs a static call graph of \(P\) and uses developer provided entry points \(E=\{e_1,\dots,e_m\}\).
    If no entry points supplied,
    \sysname heuristically selects externally invocable functions (e.g. Python callables not prefixed with an underscore) that accept at least one string parameter as entry points. 

For each \(r\), 
    \sysname traverses the call graph backwardly from call sites of \(r\) to identify all \(e\in E\) that can be reached, 
    recording all interprocedurally reachable functions \(f\) as the aggregated path \(p\) for the pair \((r,e)\).
    Within each path \(p\),
    functions \(f\) invoked before reaching \(r\) are classified into three categories:
    (i) \(f_{str}\): functions that, before reaching \(r\), accept a string parameter \(p_{str}\) and return a string or a constant value; 
    (ii) \(f_{n}\): functions that either do not take string parameters or do not return string or constant values before reaching $r$; and 
    (iii) \(f_{stk}\): functions that remain on the call stack when \(r\) is reached (i.e., invoked but not returned).

\noindent
\textbf{(2) Function Summarization via Concolic Unit Testing.}
For each $p$, 
    \sysname performs concolic unit testing to derive \emph{ReDoS-specific function summary} $\phi_f$ that capture only the features relevant to ReDoS attack string propagation for selected functions.
    
The function selection follows two rules:
(i) Function classified as \(f_{stk}\) are excluded, 
    as prior work~\cite{kim2024pbe} indicates that these functions embody the main control-flow semantics, where abstraction could compromise precision.
(ii) Functions accessing or modifying global variables are excluded,
    since their behaviors involve interprocedural interactions that unit testing cannot soundly capture.

For the selected functions,  
    those classified as \(f_n\) are assigned an empty summary (\(\phi_f = \emptyset\)) to skip them in subsequent analysis.  
For each selected \(f_{\text{str}}\),  
\sysname employs a heuristic approach that formulates four hypotheses:  
(i) $s_a \;\text{in}\; p_{\text{str}} \And ret = CONST$,  
(ii) $s_a \;\text{in}\; p_{\text{str}} \And s_a = ret$,  
(iii) $s_a = p_{\text{str}} \And s_a = ret$,  
(iv) $s_a \;\text{in}\; p_{\text{str}} \And s_a \;\text{in}\; ret$,
 where $ret$ represents the function return value and  $CONST$ represents constant value like True or False.
\sysname then performs concolic execution on each \(f_{\text{str}}\) to validate these hypotheses sequentially.  
If a counterexample is discovered for a hypothesis,  
\sysname proceeds to evaluate the next one;  
otherwise,
    the validated hypothesis is adopted as the function summary.  
If none of the hypotheses hold,
    \sysname follows traditional approaches, constructs \(\phi_f\) by disjunct all symbolic path explored during unit testing.


\noindent
\textbf{(3) System-Level Concolic Execution.}
After deriving \(\phi_f\), 
for each path \(p\), \sysname performs system-level concolic execution from the entry point \(e\).  
It instruments \(P\) with an assertion that \(s_a\) does not reach the invocation of \(r\).  
During execution, the SMT solver attempts to generate an input that falsifies this assertion while respecting the function summaries \(\phi_f\).  
If a counterexample is produced, \sysname treats it as an input that drives \(s_a\) to \(r\) under realistic execution and reports it as a refined, exploitable attack string.  
If no counterexample is found, the path is considered unreachable, and \sysname proceeds to the next path until either an exploitable attack string is found or all paths are exhausted.

\noindent
\textbf{Complexity and Correctness Discussion.}
The core principle of \sysname's ReDoS-specific compositional concolic refinement is to retain execution features relevant to attack string propagation while pruning irrelevant behaviors,
    thereby reducing summary size, computational cost, and path explosion. Following validation of the optimization proposed in prior work~\cite{godefroid2007compositional},
    we discuss the algorithm complexity and correctness of our approach.

Regarding algorithmic complexity,
    let \(P\) contain \(n\) functions, 
    each with \(b\) internal paths. 
    The DART algorithm~\cite{10.1145/1065010.1065036} may explore up to \(O(b^{n})\) interprocedural paths in the worst case. The SMART~\cite{godefroid2007compositional} mitigates this via traditional function summaries; assuming there are at most \(b\) internal paths per function, whole-program analysis explores \(O\bigl(n\cdot b\cdot 2^{x}\bigr)\) paths in the worst case, where \(x\) is the number of entries retained in each function summary. Our method reduces \(x\) by keeping only entries relevant to attack string propagation, thus lowering the overall complexity \(O\bigl(n\cdot b\cdot 2^{x}\bigr)\).

Concerning correctness, 
    \sysname does not modify the SMT solver, and every produced counterexample corresponds to a concrete execution; thus our refinement is \emph{precision-preserving relative to DART} that introduce no false positives. While we cannot eliminate false negatives entirely due to the under-approximation nature of concolic execution~\cite{vandenbogaerde2025abstracting,wang2024wacana,yun2018qsym}, our approach preserves the features relevant to attack string propagation, ensuring all feasible propagation paths remain symbolically analyzable and \emph{avoiding reductions in analyzable propagation coverage}.
A concrete example illustrating the step-by-step concolic refinement process 
is available in Appendix~\ref{appendix:concolic_example}.



\subsection{Implementation}
\vspace{-1pt}
We implement a prototype of \sysname for analyzing Python programs, 
    motivated by prior work that identifies Python as the language most vulnerable to ReDoS~\cite{davis2019rethinking}. 
\sysname is fully automated for attack string generation and concolic refinement and can be adapted to other languages by using a concolic executor for the target language, as all other components are language independent.
For detecting vulnerable regexes we use \tool{regexploit}~\cite{doyensec_regexploit_2021} to extract regexes and \tool{RENGAR}~\cite{wang2023effective} to identify vulnerable regexes as \tool{RENGAR} has been reported to achieve the best performance in prior work~\cite{wang2023effective}.
We then implement the rAST construction, vulnerable pattern identification and
    attack string generation pipeline on top of \tool{RENGAR}~\cite{wang2023effective}.
\sysname can be integrated with any ReDoS detector.
For constraint-guided refinement, 
    \sysname uses \tool{pytype}~\cite{pytype} to infer API argument type and identify entry points,
    constructs a call graph with \tool{PYCG}~\cite{salis2021pycg} followed by a custom call flow analysis,
    and builds the ReDoS-specific compositional concolic execution module on \tool{CrossHair}~\cite{crosshair}, a Z3~\cite{z3}-based concolic execution framework.
    We employ \tool{CrossHair} for its active maintenance and robust Python support, while our algorithm remains compatible with any concolic execution backend.
All the development, implementation and evaluation were conducted on Debian GNU/Linux~12, 32GB RAM, 1TB SSD.
\vspace{-0.5em}

%% file: Tex/experiments.tex
\vspace{-1em}
We evaluate the performance of \sysname  by answering the following research questions (RQs):

\noindent\hangindent=0.9em  \textbullet\ \textbf{RQ1 (Efficiency):} How efficient is \sysname at producing attack strings?

\noindent\hangindent=0.9em  \textbullet\ \textbf{RQ2 (Effectiveness):} How effective is \sysname at reproducing CVE-disclosed ReDoS vulnerabilities?

\noindent\hangindent=0.9em  \textbullet\ \textbf{RQ3 (Overhead):} Can \sysname exploit previously undisclosed ReDoS in real-world projects?

During the experiment, we use regex matching time to measure both attack success and efficiency.
To ensure stability, all reported results are averaged over three independent runs.
Additionally, 
    following common practice in program-analysis research~\cite{wang2019oo7,wei2023compiling}, attack string generator runtime overhead is evaluated by the execution time of each tool to reflect the computational cost of the analysis.

\vspace{-1em}
\subsection{Experiment Setup}
\vspace{-1em}
We create three datasets, a \emph{baseline dataset}, a \emph{CVE dataset} and a \emph{real-world project dataset}, 
    to separately assess the efficiency and real-world effectiveness of \sysname. 

\noindent
\textbf{Baseline dataset.} This dataset contains confirmed vulnerable regexes collected from real-world programs, 
    which are used for assessing whether \sysname can produce attack strings. 
In particular, we first collected two widely used regex corpora, \tool{Corpus}~\cite{chapman2016exploring} and \tool{RENGAR} dataset~\cite{rengarweb}, containing 131,701 regexes collected from real-world programs.  
We then applied five ReDoS detectors \tool{RENGAR}~\cite{wang2023effective}, \tool{ReScue}~\cite{shen2018rescue}, \tool{Regulator}~\cite{mclaughlin2022regulator}, \tool{Revealer}~\cite{liu2021revealer}, \tool{ReDoSHunter}~\cite{li2021redoshunter} with 100\% reported precision on these two corpora to detect the vulnerable regexes.
Their outputs are verified by measuring regex matching time, and eventually we confirmed \textbf{17,962 vulnerable regexes}.



\noindent
\textbf{CVE dataset.} This dataset contains the ReDoS-related CVEs to evaluate whether the attack strings generated by \sysname can reproduce known ReDoS attacks. 
Specifically, we collected ReDoS instances from the National Vulnerability Database~\cite{cve} (i.e., CVE-disclosed vulnerabilities).
In total, 
    to capture recent vulnerability trends and given that \sysname currently targets Python,
    we collected all 51 ReDoS-related CVEs reported in Python projects over the past five years.
    After excluding those without publicly available attack strings, 
     \textbf{31 CVEs} remained for our final analysis, 
     each providing verified attack strings that trigger the corresponding vulnerable regexes.

\noindent
\textbf{Real-world project dataset.} This dataset contains previously undisclosed ReDoS vulnerabilities whose corresponding vulnerable regexes were automatically detected by the ReDoS detectors. 
These vulnerabilities had not been publicly reported and lacked known attack strings at the time of collection. 
They are used to evaluate \sysname’s capability to synthesize new exploits.
To construct this dataset,
    we selected Python  projects  with at least 200 GitHub stars and 10k monthly downloads, 
    and then randomly chose 100 projects from the top 1,000 by downloads.
After manually reviewing their documentation, we excluded those that do not process any external user inputs (e.g., package management tools), leaving 74 candidates.
We further examined their codebases and discarded projects without regex usage, resulting in 38 remaining projects.
Finally, we applied the same five ReDoS detectors used in the baseline dataset and identified 12 projects containing vulnerable regexes.
These projects span database frameworks, LLM agents, text processing, and data analysis, with code sizes ranging from 7,860 to 443,427 lines.
In total, 
    \textbf{95 real-world vulnerable regexes} are included. 
Detailed statistics are provided in Appendix~\ref{apx:project}.

\subsection{Efficiency (RQ1)}
To evaluate attack string generation performance, we compare \sysname with \tool{RENGAR}~\cite{wang2023effective}, a state-of-the-art hybrid analysis tool.
Since \tool{RENGAR} has been shown to outperform existing approaches,
we use it as our primary benchmark.
We exclude \tool{EVILSTRGEN}~\cite{su2024towards} from comparison as it targets non-backtracking engines and is designed to generate attack strings of at least length $k$~\cite{su2024towards}, which is orthogonal to our goal of minimizing attack string length for backtracking engines.
The experiment is carried out Python’s default regex engine \tool{re}~\cite{python_re} for reflecting real-world deployment settings.
For each vulnerable regex in the benchmark dataset, 
    we ran \sysname and \tool{RENGAR} independently to produce attack strings. For every (regex, tool) pair, we then determined the minimum attack string length required to push the regex match time beyond a specified threshold.
Following the settings in prior usability studies~\cite{10.5555/2821575}, 
    we set the thresholds to 0.1s, 1s, and 10s, which correspond to the three human-perceptual limits for instantaneous, uninterrupted, and sustained response times.

Specifically,
    given a dataset $D$ containing $N$ regexes,
    for regex $r_i \in D, i \in [0,N]$, 
    let $L_{R}$ and $L_{P}$ be the minimal attack string lengths generated by \tool{RENGAR} and \sysname to reach matching time cost threshold $T$.  
We use the following two metrics to evaluate the efficiency.
\ding{182} We define $R_{T}^i$ as the proportional increase in length required by \tool{RENGAR} over \sysname to reach threshold $T$ when attacking $r_i$, 
    and $\bar{R_{T}}$ denotes the mean difference across dataset $D$.
Formally we have  $R_{T}^i = \frac{L_{R} - L_{P}}{L_{P}}$ and $\bar{R_{T}} = \tfrac{1}{N}\!\sum^N_{i=0} R_{T}^i$.
A positive $\bar{R_{T}}$ indicates that \tool{RENGAR} requires longer inputs to reach the threshold, 
whereas a negative value implies that \sysname does.
\ding{183} 
We then compute the percentage $P_T$ for which $L_R > L_P$, in order to compare the relative effectiveness of \tool{RENGAR} and \sysname in reaching the threshold $T$ across the entire dataset.
To avoid trivial variations, cases with relative deviation 
$|L_{R} - L_{P}| / \min(L_{R}, L_{P}) < 5\%$ 
are regarded as equivalent and excluded from the count.
\input{Tex/Table/baseline}

\noindent
The results are listed in Table~\ref{tab:baseline_avg}.
When setting the threshold to 0.1s, 
    \sysname outperforms \tool{RENGAR} in 92.8\% of the cases, with \tool{RENGAR} requiring attack strings that are on average 97.2× (times) longer to reach the same threshold.
Such performance remains consistent as the threshold increases.
At 1s, 
   \sysname performs better in 92.7\% of the cases, with \tool{RENGAR} needing strings that are 1128.5× longer;
At 10s, the corresponding values are 93.7\% and 3872.4, 
    respectively.
While considering the runtime overhead of string generation, 
    \tool{RENGAR} averagely takes 2.47 seconds to generate an attack string, whereas \sysname requires 5.11 seconds due to its rAST and rPath analyses for deeper structural exploration.
    
\begin{figure*}[htbp]
    \centering
    \includegraphics[width=0.95\textwidth]{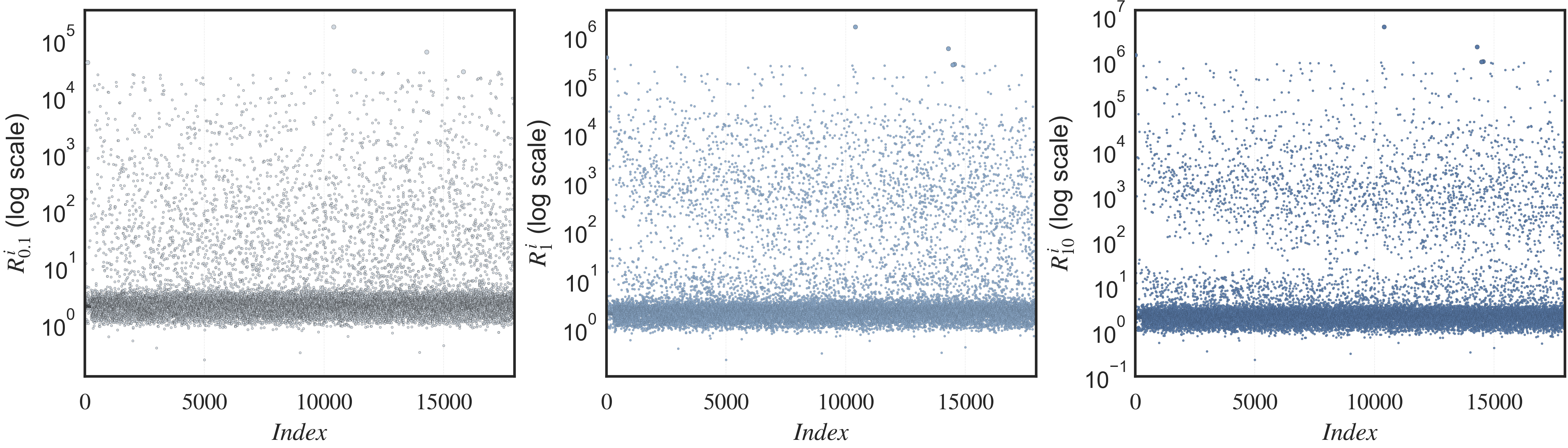}
    \caption{Distribution of $\bar{R}_T$ on baseline. The extreme values reach $10^{5}$, $10^{6}$, and $10^{7}$ at $T=0.1$, $1$, and $10$, respectively. 
    \label{fig:result_rq1}}
\end{figure*}
\noindent
\textbf{Scrutiny of Results.}
We firstly study the distribution of $R^i_{T}$ across the baseline dataset.
As shown in Figure~\ref{fig:result_rq1},
    for the majority of regexes, 
    \tool{RENGAR} requires roughly 3x longer attack strings to reach a given threshold compared with \sysname.
Nevertheless,
    \sysname can frequently produce exceptionally efficient attack strings,
    with $R^i_{T}$ reaching between $10^{4}$ and $10^{7}$ as the threshold increases,
    showing that \sysname can expose weaknesses in regexes and trigger far more severe ReDoS behaviors.

To scrutinize the observations,
    we then randomly selected 300 regexes from the baseline dataset and manually examined their results (see Table~\ref{tab:scrutinize_results}). 
Among these, 5 (1.6\%) regexes follow the EOL pattern, 133 (44.3\%) the LIL pattern, and 162 (54.0\%) the PML pattern. 
Across all thresholds,
    \tool{RENGAR} outperforms \sysname in 9 cases, while \sysname surpasses \tool{RENGAR} in 183 cases. 
The remaining 108 cases with results differences within 5\% are omitted as comparable performance.

For regexes with \underline{EOL patterns}, \sysname outperforms in two cases and achieves comparable results in three, with \tool{RENGAR} requiring less than 6\% longer attack strings to reach the same matching cost.
For \underline{PML patterns}, \sysname excels in 57 cases and comparably in 105, where \tool{RENGAR} needs about twice the string length to reach the same threshold.
These improvements stem from \sysname’s selection of the shortest feasible repetition unit, which increases the stop point number and matching cost.
Although the gain for EOLs appears modest, it is substantial given the exponential sensitivity to input length.
For \underline{LIL patterns}, 
    \sysname outperforms in 124 cases,
    achieving 162.6$\times$, 1485.4$\times$, and 4805.8$\times$ shorter strings than \tool{RENGAR} 
    at thresholds of $T=0.1$, $T=1$, and $T=10$, respectively.
Among all pattern classes, LIL contributes the most to the overall $\bar{R}_{T}$.
\sysname’s advantage grows as the threshold increases as its effective separated stop point handling 
    and shortest repetition unit selection induce non-linear backtracking growth with input length.
\sysname underperforms in 9 cases: 
one due to a distinctive regex structure leads the engine to prematurely accept the attack string when stop points are ordered separately,
and eight involving \func{LE}s that consume limited input (e.g., with quantifier \{0,20\}), violating Heuristic~\ref{ass:1}’s assumption that \func{LE}s can consume sufficiently long substrings.
However, 
    defining a precise threshold for ``sufficient consumption length'' is inherently difficult, 
    as it depends heavily on the structural complexity of each regex. 
Since such cases constitute only 2.7\% of the total, 
    we leave this issue for future investigation.

\noindent
\textbf{Evaluation across Different Regex Engines.}
We assess \sysname’s generality by comparing it with \tool{RENGAR} across multiple regex engines in different languages, 
as these engines evolve independently with varying defenses and optimizations against ReDoS, and their differences in supported regex features and runtime behaviors may also affect attack performance.
For instance, Java introduced a bounded caching mechanism in Java 9 to reduce backtracking, while JavaScript added a non-backtracking engine for specific patterns in 2020 \cite{v8nonbacktracking,bhuiyan2025sok,jdk}.
Hence, 
    we further include four widely used engines reported to be vulnerable \cite{bhuiyan2025sok}: Python’s \tool{regex} \cite{pythonregex}, JavaScript’s \tool{RegExp} \cite{nodejs} with the experimental ReDoS defense enabled, and Java’s \tool{java.util.regex} \cite{javautilregex} on Java 8 and Java 23 to observe the impact of engine evolution. 
All our experiments use the official releases as of October 2025.

The results are shown in Appendix~\ref{app:cross}.
\sysname consistently outperforms \tool{RENGAR} across all regex engines, with \tool{RENGAR} requiring attack strings 12.1–5051.8× longer to reach the same matching cost threshold on the baseline dataset.
Among pattern categories,
    LIL regexes yield the largest improvement,
    making \tool{RENGAR} requiring 14.9–5510.1× longer strings to achieve a threshold, whereas regexes with PML patterns requiring roughly 20\%–180\% longer strings.
The improvement, however, varies across engines.
\sysname achieves its greatest advantage on Java-8, 
    whose traditional backtracking-based implementation remains highly vulnerable.
In contrast, the advantage diminishes on Java-23 and Node.js-25,
    where the bounded caching mechanism and non-backtracking optimization partially prunes redundant backtracking paths, 
    thereby reducing the matching costs of attack strings generated by both tools toward near-linear behavior and consequently lowering the observed $\bar{R}_{T}$ and $P_{T}$.
For Python’s \tool{regex} engine, \sysname performs comparably to official \tool{re} engine, as both engine share the same backtracking NFA core.

\subsection{CVE Reproduction (RQ2)}
To verify the effectiveness of the attack strings in reproducing ReDoS attacks, 
    we utilize the CVE dataset and re-execute the attacks demonstrated in each CVE report. 
For comparison, 
    we adopt \tool{RENGAR} integrated with \tool{CrossHair},
    extended with a traditional compositional concolic execution module, as the primary baseline.
We also compare the attack strings generated by \sysname with the CVE disclosed attack strings. 
Consistent with prior studies~\cite{bhuiyan2025sok,davis2018impact,mclaughlin2022regulator,su2024towards},
    a ReDoS exploit is considered successful if it causes a system hang of $>10s$ using an input smaller than 100KB.

As illustrated in Table~\ref{tab:vuln_known}.
    \sysname successfully reproduced 30 (96.8\%) ReDoS attacks whereas \tool{RENGAR} reproduced 24 (77.4\%). 
We manually inspected the only CVE that neither tool reproduced and found that it requires complex user interactions and the construction of highly specific input structures. 
As the concolic executor does not fully model these interactions, the exploitation failed. 
Among the six cases where \tool{RENGAR} failed but \sysname succeeded, four were caused by complex path conditions involving regex operations, leading to path explosion in traditional compositional concolic execution. 
\sysname uses ReDoS-specific function summaries to eliminate unnecessary behavior tracking, thereby reducing path exploration overhead.
The remaining two failures resulted from the inability of \tool{RENGAR} to generate sufficiently effective attack strings. 
\sysname instead identified that these two cases have an LIL pattern and generated attack strings that were thousands times more efficient and successfully reproduced the ReDoS attack.
Notably, three CVE-disclosed attack strings failed to reproduce the attack due to excessive length, supporting prior observations that some reported attacks rely on unrealistic input sizes~\cite{bhuiyan2025sok}. Two of these cases exhibit LIL patterns for which \sysname synthesized far more effective attack strings. The remaining case involves an EOL pattern: the disclosed string used an incorrect repetition unit and omitted a mandatory character, so the regex rejected the input without any backtracking.

For runtime overhead, 
    on average, as shown in Table~\ref{tab:vuln_known}, with ReDoS-specific function summaries, \sysname reduces CVE reproduction time by 2.7 minutes per case.
\input{Tex/Table/vuln}
    
\subsection{Real-world Effectiveness (RQ3)}
We further apply \sysname to the real-world project dataset to assess its ability to detect previously undisclosed vulnerabilities.
The results of exploiting undisclosed vulnerabilities are listed in Table~\ref{tab:vuln_unknown}.
In detail,
    \sysname successfully triggered 59 out of 95 vulnerable regexes in 12 projects, while \tool{RENGAR + Crosshair} triggered 34 of them. 
The ablation variant,
    retaining our ReDoS-specific concolic refinement but replacing \sysname's attack string generation with \tool{RENGAR}, triggered 45 ReDoS vulnerabilities.
Among them,
    21 vulnerabilities in six projects have been confirmed by developers.

After manually inspecting the 36 vulnerable regexes that neither \sysname nor \tool{RENGAR} could exploit, 
    we found:

\noindent\hangindent=0.9em  \textbullet\ 30 cases cannot be triggered at all, namely, these are the false positives caused by the ReDoS vulnerability detectors that we integrated. 
Among those 30 cases, 
    21 of them involve internal regexes that are unreachable from any program entry point and are not exposed to user input.
For instance, 
    \tool{numpy} contains five vulnerable regexes that are used only to parse CPU version information. 
Nine other failures result from security filters that prevent the manifestation of ReDoS behavior. For example, \tool{python-markdown} adopts the \tool{re2}~\cite{google-re2} engine in specific modules, which enforces linear-time regex matching and eliminates backtracking-based ReDoS.

\noindent\hangindent=0.9em  \textbullet\ The remaining six failed cases have either complex input structures or indirect dataflows that concolic execution cannot model.
For example, in \tool{pydantic} a vulnerable regex stored as an object attribute is repeatedly retrieved and reused across components, breaking concolic constraint propagation and causing the analysis to fail.

Of the 25 cases where only \tool{RENGAR} failed, we found 11 cases involves complex cross-function regex operations, where summarizing function behavior led to excessively large summaries in traditional compositional concolic analyzers;  these are resolved by \sysname and the ablation variant using ReDoS-specific summaries.
The remaining 14 result from \tool{RENGAR}’s failure to generate sufficiently efficient attack strings within length limits.

Such performance varies across projects. \sysname exhibits greater advantages as the complexity of input processing preceding regex invocation increases.
For example,
    in \tool{pygments},
    a syntax-highlighting library that applies regexes directly to user input, 
    \tool{RENGAR} can easily trace regex usage and thus performs comparably to \sysname.
In contrast,
    \tool{python-markdown} preprocesses input by splitting text blocks and metadata before applying regex parsing, creating deep call chains that hinder \tool{RENGAR} from identifying vulnerable paths, whereas \sysname effectively resolves them with ReDoS-specific function summaries.
Manual analysis further shows that, as presented in Appendix~\ref{apx:project}, an average of 32\% of functions across all projects benefit from  ReDoS-specific summaries, with variation driven by developers' coding style like the usage of global variables.
    
Table~\ref{tab:vuln_unknown} reports runtime overhead. \sysname requires on average 14.6 minutes per project versus 16.1 minutes for the ablation variant and 18.9 minutes for \tool{RENGAR} with \tool{CrossHair}, a reduction of 9.3\% and 22.8\% respectively. 
With ReDoS-specific function summaries and efficient attack string generation,
    \sysname propagates constraints across functions with lower overhead and exploits real-world ReDoS within practical input length limits.
More detailed runtime breakdown can be found in Appendix~\ref{apx:runtime}.


\subsection{Threats to Validity} 
\noindent
\textbf{External Validity.} 
Our design is grounded in formal analysis–based reasoning and is inherently language-agnostic.
While we have evaluated it across multiple regex engines, the experiments primarily focus on the Python ecosystem.
In addition,
    practical regex engines may support extended features such as lookarounds and named groups~\cite{friedl2006mastering}. 
Named groups and non-capturing groups do not affect \sysname as they preserve the order and overlap structure of loop expressions;
lookarounds may affect \sysname in specific cases by altering the valid alphabet of adjacent loop expressions via zero-width assertions.
Nevertheless,
    empirical evaluation results show that \sysname remains effective across diverse real-world regexes, demonstrating strong robustness beyond its formal assumptions.
Extending the formal analysis and evaluation to support advanced regex features, additional programming languages, and runtime environments is left for future work.
    
\noindent
\textbf{Internal Validity.} 
In regex matching, full match semantics require the entire input string to satisfy the regex, while partial match semantics succeed when any substring matches the regex.
The majority of previous works~\cite{wang2023effective,li2021redoshunter,petsios2017slowfuzz,li2022regexscalpel,shen2018rescue} adopt the full match setting, and \sysname follows it for fair comparison and easier integration.
We leave the support for partial match semantics as future work.

For the baseline tool,
    since no open-source compositional concolic execution framework exists for Python, 
    we implemented the baseline using a conventional SMART algorithm~\cite{godefroid2007compositional,kim2019target} following the standard methodology and employed the state-of-the-art concolic executor \tool{Crosshair} to ensure fairness and implementation consistency. 
Besides,
    although taking execution time as the metric may introduce uncertainty, 
    it is the prevailing practice in prior works~\cite{bhuiyan2025sok,davis2018impact,mclaughlin2022regulator,su2024towards}; each measurement was therefore averaged over three independent runs to reduce variance. 
While counting regex matching steps could be more precise, 
    no formally established criterion delineates the threshold of matching steps required to characterize a ReDoS manifestation,
    and most engines lack the required instrumentation, 
    making execution time the most comparable and reproducible measure.

%% file: Tex/Table/baseline.tex
\begin{table}[h!]
\centering
\caption{Comparison between \sysname and \tool{RENGAR}}
\begin{tabular}{
>{\centering\arraybackslash}p{1.3cm}|
>{\centering\arraybackslash}p{1.5cm}|
>{\centering\arraybackslash}p{1.5cm}
}
\toprule[1.5pt]
\textbf{Threshold} & \textbf{$\bar{R}_T$} & \textbf{$P_T$} \\
\midrule
$T = 0.1$ & 97.2 & 92.8\% \\
$T = 1$   & 1128.5 & 92.7\% \\
$T = 10$  & 3872.4 & 93.7\% \\
\bottomrule[1.5pt]
\end{tabular}
\label{tab:baseline_avg}
\end{table}

\vspace{-1em}

\begin{table}[h!]
\centering
\caption{Results of pattern study}
\begin{tabular}{c|cc|cc|cc}
\toprule[1.5pt]
\raisebox{-1.0ex}[0pt][0pt]{\textbf{Pattern}} 
& \multicolumn{2}{c|}{$T = 0.1$} 
& \multicolumn{2}{c|}{$T = 1$} 
& \multicolumn{2}{c}{$T = 10$} \\
\cmidrule(lr){2-3} \cmidrule(lr){4-5} \cmidrule(lr){6-7}
& \textbf{$\bar{R}_T$} & \textbf{$P_T$} & \textbf{$\bar{R}_T$} & \textbf{$P_T$} & \textbf{$\bar{R}_T$} & \textbf{$P_T$} \\
\midrule
EOL & 0.06 & 100\% & 0.01 & 100\% & 0.01 & 100\% \\
LIL & 162.6 & 94.1\% & 1485.4 & 93.7\% & 4805.8 & 93.7\% \\
PML & 1.7 & 100\% & 1.3 & 100\% & 1.5 & 100\% \\
\bottomrule[1.5pt]
\end{tabular}
\label{tab:scrutinize_results}
\end{table}

%% file: Tex/Table/vuln.tex
\begin{table}[h!]
\centering
\caption{Comparison results on ReDoS CVEs.}
\label{tab:vuln_known}
\scriptsize
\setlength{\tabcolsep}{6pt}
\renewcommand{\arraystretch}{0.95}

\begin{tabular}{p{2.0 cm}|>{\centering\arraybackslash}p{1.1cm}>{\centering\arraybackslash}p{1.1cm}>{\centering\arraybackslash}p{1.1cm}}
\toprule
 & \textbf{\sysname} & \textbf{RENGAR} & \textbf{Disclosed} \\ 
\midrule
\textbf{Exploit Rate (\%)} & 96.8 & 77.4 & 90.3 \\  
\textbf{Avg. Time (min)} & 5.3 & 8.0 & -- \\  
\bottomrule
\end{tabular}\\[2pt]
{\footnotesize $^*$Disclosed: attack strings documented in official CVE reports.}
\end{table}

\vspace{-1em}

\begin{table}[h!]
\centering
\caption{Comparison results on undisclosed ReDoS.}
\label{tab:vuln_unknown}
\scriptsize
\setlength{\tabcolsep}{4pt}
\renewcommand{\arraystretch}{0.95}
\begin{tabular}{p{2.0cm}|
  >{\centering\arraybackslash}p{0.4cm}
  >{\centering\arraybackslash}p{1.1cm}|
  >{\centering\arraybackslash}p{0.4cm}
  >{\centering\arraybackslash}p{1.1cm}|
  >{\centering\arraybackslash}p{0.4cm}
  >{\centering\arraybackslash}p{1.1cm}}
\toprule
\textbf{Target Project}
  & \multicolumn{2}{c|}{\textbf{\sysname}}
  & \multicolumn{2}{c|}{\textbf{Ablation\textsuperscript{2}}}
  & \multicolumn{2}{c}{\textbf{RENGAR}} \\
\cmidrule(lr){2-3}\cmidrule(lr){4-5}\cmidrule(lr){6-7}
 & \textbf{ReDoS} & \textbf{Time\textsuperscript{1}} & \textbf{ReDoS} & \textbf{Time} & \textbf{ReDoS} & \textbf{Time} \\
\midrule
nltk                & 3  & 23.0 & 2  & 26.8 & 1  & 28.3 \\
python-markdown     & 7  & 4.0  & 4  & 4.3  & 1  & 5.2  \\
pygments            & 15 & 14.6 & 14 & 19.1 & 14 & 21.4 \\
sphinx              & 5  & 19.1 & 5  & 22.7 & 5  & 25.4 \\
peewee              & 1  & 9.3  & 0  & 12.1 & 0  & 12.8 \\
readme-ai           & 8  & 4.0  & 5  & 4.6  & 2  & 6.1  \\
numpy               & 7  & 41.9 & 6  & 45.3 & 4  & 47.4 \\
h5py                & 3  & 12.3 & 2  & 13.4 & 2  & 16.3 \\
langchain           & 5  & 30.9 & 3  & 34.8 & 2  & 35.9 \\
neural-compressor   & 1  & 5.1  & 1  & 6.2  & 0  & 9.3  \\
pydantic            & 3  & 7.8  & 2  & 11.7 & 2  & 13.0 \\
cffi                & 1  & 3.2  & 1  & 5.6  & 1  & 6.3  \\
\midrule
\textbf{Total} & \textbf{59} & --            & \textbf{45} & --            & \textbf{34} & -- \\
\textbf{Avg.}  & --          & \textbf{14.6} & --          & \textbf{16.1} & --          & \textbf{18.9} \\
\bottomrule
\end{tabular}\\[3pt]
{\footnotesize 1. Time in minutes. 2. Ablation: RENGAR string generation + \sysname concolic refinement.}
\end{table}

%% file: Tex/relatedWork.tex
\noindent
\textbf{ReDoS Attack String Generation.}
Traditional ReDoS attack-string generation is coupled with detection: 
    once a regex is flagged as vulnerable, generators synthesize inputs that induce catastrophic backtracking.
Prior work implements this via static, dynamic, and hybrid techniques.

Early tools such as \tool{regexploit}~\cite{doyensec_regexploit_2025} and \tool{redos-detector}~\cite{kkos_oniguruma_2025} generate attack strings using syntax-level rules for common patterns (e.g., infinite or nested loops), 
    and Demoulin et al.~\cite{demoulin2019detecting} scale this idea with machine learning.
While lightweight, 
    these pattern-based approaches lack semantic reasoning about match execution and thus provide limited exploitability guarantees.
Recent work innovatively demonstrates that compressed data direct computing can be integrated with homomorphic encryption to provide efficient and secure computation~\cite{zhang2021tadoc,hoco}.
Automata-based methods instead construct and analyse the automaton for a regex to detect worst-case behaviours:     
    \tool{Rexploiter}~\cite{wustholz2017static} synthesizes attacks via adversarial automata construction, Weideman et al.~\cite{weideman2016analyzing} speed up analysis with a prioritized NFA, and recent works~\cite{turovnova2022counting,su2024towards} leverage DFA techniques to produce attack strings for non-backtracking engines.
However, 
    their abstractions can still over- or under-approximate runtime behavior, and thus cannot guarantee practical accuracy.

Dynamic exploration approaches solve this problem by leveraging feedback-driven exploration of regex execution to synthesize inputs that incur high matching cost.
\tool{SDLFuzzer}~\cite{kirrage2013static} performs this exploration using randomly generated strings, but its unguided search leads to low efficiency and limited coverage.
\tool{Regulator}~\cite{mclaughlin2022regulator} and \tool{SlowFuzz} improve upon this by adopting grey-box fuzzing, instrumenting the regex engine to guide input mutation through execution feedback.
\tool{ReScue}~\cite{shen2018rescue} employs a similar strategy, using an evolutionary genetic algorithm to evolve inputs toward higher cost.
Nevertheless, these methods remain constrained by limited coverage and high time overhead.

Hybrid analysis tools combine static attack string generation with dynamic validation to balance accuracy and overhead.
\tool{ReDoSHunter}~\cite{li2021redoshunter} integrates pattern-guided static detection with dynamic validation, forming the first hybrid ReDoS analysis framework.
\tool{NFAA}~\cite{weideman2016analyzing} and \tool{Revealer}~\cite{liu2021revealer} extend NFA-based static analyses using dynamic testing to confirm exploitability.
\tool{RENGAR}~\cite{wang2023effective} further improves scalability by mitigating disturbances from subregexes, while \tool{Badger}~\cite{noller2018badger} integrates fuzzing with symbolic execution.
Despite runtime overhead, hybrid analysis consistently delivers the best experimental performance.
Yet, these methods still prioritize generating successful attack strings rather than optimizing attack efficiency or validating real-world exploitability, limiting their practical applicability.

\noindent
\textbf{Compositional Vulnerability Analysis.}
Function summaries have been widely adopted in vulnerability analysis.
Early approaches aimed to summarize all behaviors of each function in the program. 
    SMART~\cite{godefroid2007compositional} disjoined symbolic execution results over inputs and outputs, and Anand et al.~\cite{anand2008demand} extended this with on-demand first-order formula refinement. FOCAL~\cite{kim2019target} applied counterexample-guided refinement in unit testing, while FunFrog~\cite{sery2011interpolation} and HiFrog~\cite{alt2017hifrog} built SSA-form summaries refined iteratively via counterexamples.
While accurate, 
    these methods often incur substantial summary construction and reuse overhead, 
    which can lead to path explosion.
Recent works address this through selective abstraction: generating summaries only for specific components~\cite{sato2013towards,kim2024pbe}, abstracting array operations~\cite{armando2014counterexample} or global variable statements~\cite{kim2019model}, and extending bounded model checking to concurrent code via scheduling summaries~\cite{yin2018scheduling}. Most closely aligned with our approach, NLP-eye~\cite{wang2019nlp} and Goshawk~\cite{lyu2022goshawk} restrict summaries to memory allocation and deallocation semantics to reduce size and accelerate analysis.

%% file: Tex/conclusion.tex
In this paper, 
    we present \sysname, a hybrid framework that, guided by formal analysis, generates efficient attack strings that amplify matching cost.
To guarantee real world exploitability,
    it incorporates a ReDoS-specific compositional concolic execution to refines and validates each attack string at program level.
In evaluation, \sysname generates attack strings 97.2× to 3872.4× shorter while achieving equivalent matching cost with \tool{RENGAR}.
It reproduced 96.8\%  of the ReDoS CVE attacks and exploited 59 previously undisclosed ReDoS vulnerabilities in 12 projects, 
    outperforming attack strings produced by state-of-the-art tools and those disclosed in CVE reports.
We release \sysname at~\cite{pufferdos}.

%% file: Tex/appendices.tex
\subsection{Formal Proofs\label{sec:proof}}
\label{app:proof}

\noindent
\textbf{Proof of Lemma~\ref{lem:1}.}
By Definition~\ref{def:repunit}, each boundary between consecutive copies of a feasible repetition unit $y$ inside the infix is a reachable stop point for every participating \func{LE} that will be considered by the backtracking matcher. For an infix of length at most $n$, the number of such boundaries
\[
m(y) = \Bigl\lfloor \frac{n}{|y|} \Bigr\rfloor.
\]
Therefore, if $|y_1|<|y_2|$ then $m(y_1)\ge m(y_2)$, with strict inequality whenever $\bigl\lfloor n/|y_1| \bigr\rfloor > \bigl\lfloor n/|y_2| \bigr\rfloor$. Let $s_a^{(1)} = x\,y_1^{q_1}\,z$ and $s_a^{(2)} = x\,y_2^{q_2}\,z$ be the corresponding attack strings with $|s_a^{(i)}| = n$. Under Heuristic~\ref{ass:1} (matching cost is non-decreasing in the number of reachable stop points), it follows:
\[
C_r\bigl(s_a^{(1)},y_1,n,m(y_1)\bigr)\;\ge\; C_r\bigl(s_a^{(2)},y_2,n,m(y_2)\bigr),
\]
with strict inequality when $m(y_1)>m(y_2)$. 

\noindent
\textbf{Proof of Lemma~\ref{lem:2}.}
Let $s_a^{(y)}=x\,y^{q}\,z$ with $|s_a^{(y)}|= n$. For a given $y$, let $m_i(y)$ and $m_{i+1}(y)$ denote the numbers of reachable stop points contributed by $\func{LE}_i$ and $\func{LE}_{i+1}$ on $s_a^{(y)}$ to the left of the failing location. Because every symbol of $y_1$ lies in $\Sigma(\func{LE}_i)\cap\Sigma(\func{LE}_{i+1})$, repeating $y_1$ creates stop points for both loops at the boundaries between consecutive copies, so $m_i(y_1)>0$ and $m_{i+1}(y_1)>0$ grow with $q$. In contrast, by construction $y_2$ is consumable by exactly one loop, hence on the repeated segment only one of $m_i(y_2),m_{i+1}(y_2)$ increases while the other remains zero.

Therefore the total number of reachable stop points with $y_1$ satisfies
\[
 m_i(y_1)+m_{i+1}(y_1)\;\ge\; m_i(y_2)+m_{i+1}(y_2),
\]
with a strict inequality whenever the second loop contributes at least one stop point on the $y_1$ segment. By Heuristic~\ref{ass:1} (matching cost is non-decreasing in the number of reachable stop points and strictly increases when that number increases), we obtain
\[
\begin{split}
 C_r\bigl(s_a^{(y_1)},y_1,n,m_i(y_1)+m_{i+1}(y_1)\bigr)\\
 \ge\; C_r\bigl(s_a^{(y_2)},y_2,n,m_i(y_2)+m_{i+1}(y_2)\bigr),
\end{split}
\]
with strict inequality when the total number of reachable stop points increases.

\noindent
\textbf{Proof of Lemma~\ref{lem:3}.}
Consider two attack strings $s_a^{(1)}$ and $s_a^{(2)}$ of equal length $n$  that differ only in the arrangement of stop points of $\func{LE}_i$ and $\func{LE}_{i+1}$. 
Let $\func{LE}_i$ and $\func{LE}_{i+1}$ have $j$ and $k$ stop points on both  attack strings, respectively, denoted by
\[
\mathcal{P}_i = \{p_{i,0}, \dots, p_{i,j-1}\}, \quad
\mathcal{P}_{i+1} = \{p_{i+1,0}, \dots, p_{i+1,k-1}\}.
\]
In $s_a^{(1)}$, the stop points are \emph{interleaved}:
\[
s_a^{(1)}: p_{i,0} \!<\! p_{i+1,0} \!<\! p_{i,1} \!<\! p_{i+1,1} \!<\! \cdots \!<\! p_{i,j-1} \!<\! p_{i+1,k-1}
\]
whereas in $s_a^{(2)}$, they are arranged in \emph{separated} order:
\[
s_a^{(2)}: p_{i,0} < \cdots < p_{i,j-1} < p_{i+1,0} < \cdots < p_{i+1,k-1}
\]
Each time $\func{LE}_i$ backtracks at $p_{i,l}\!\in\!\mathcal{P}_i$, 
$\func{LE}_{i+1}$ backtracks at each its stop points from $p_{i,l}$ to the end of attack string. 
Let 
$M(s_a,\!\func{LE}_{i+1},\!p_{i,l}\!\!\rightarrow\!n)$ 
denote the number of stop points of $\func{LE}_{i+1}$ reachable from $p_{i,l}$ to the end of attack string, 
when stop points are interleaved as in $s_a^{(1)}$, 
some of $\func{LE}_{i+1}$’s stop points lie before  $p_{i,l}$ and are unreachable during backtracking, 
so $M(s_a^{(1)},\func{LE}_{i+1},p_{i,l}\!\rightarrow\!n) \leq k$. 
In the separated case $s_a^{(2)}$, all $k$ stop points remain reachable, i.e., $M(s_a^{(2)},\func{LE}_{i+1},p_{i,l}\!\rightarrow\!n)=k$. 
Hence the total number of reachable stop points satisfies
\[
\begin{split}
m^{(2)}=\sum_{l=0}^{j-1} M(s_a^{(2)},\func{LE}_{i+1},p_{i,l}\!\rightarrow\!n)
\\
\geq m^{(1)}=\sum_{l=0}^{j-1} M(s_a^{(1)},\func{LE}_{i+1},p_{i,l}\!\rightarrow\!n)
\end{split}
\]
As $\func{LE}_i$ backtracks over all its $k$ own stop points in both cases, its cost is constant; 
thus the total matching cost depends only on $\func{LE}_{i+1}$. 
By Heuristic~\ref{ass:1} (matching cost increases with the number of reachable stop points), we obtain
\[
C_r(s_a^{(2)},y,n,m^{(2)}) \geq C_r(s_a^{(1)},y,n,m^{(1)})
\]
so that the separated arrangement incurs no lower matching cost than the interleaved one under the same length budget.

\subsection{Regex Vulnerable Pattern Coverage}
\label{apx:pattern}
\input{Tex/Table/pattern}
Table~\ref{tab:regex_patterns} summarizes how our designed patterns cover the representative vulnerable structures reported in prior work~\cite{wang2023effective,li2021redoshunter}.
Prior studies~\cite{wang2023effective} have shown that these categories collectively subsume all vulnerable pattern types identified in existing literature.
\input{Tex/Table/realworld}

\subsection{Project Statistics}
\label{apx:project}
Table~\ref{tab:project} summarizes statistics for the 12 real-world Python projects used in our evaluation, reporting each project's GitHub stars, LOC, number of defined functions (Func), early-returning functions (\#$f$), ReDoS-specific summarizable functions ($\phi_f$), and vulnerable regexes (Vul.~$r$). In total, the dataset comprises 95 vulnerable regexes, 68,021 defined functions, and 5,848 early-returning functions, of which 1,895 (32\%) are ReDoS-specific.

\subsection{Cross-engine Evaluation Results}
\label{app:cross}
\input{Tex/Table/diff_engine}
Table~\ref{tab:cross_engine_eval} compares \sysname and \tool{RENGAR} on the baseline dataset, where the first column lists the regex engines and the subsequent columns show results under different matching cost thresholds.
Tables~\ref{tab:pattern_cross_engine}(a)–(c) present results on 300 selected regexes grouped by pattern type: (a) EOL, (b) LIL, and (c) PML.

\vspace{-1em}
\subsection{Runtime Breakdown.}
\label{apx:runtime}
Table~\ref{tab:time_breakdown} shows the runtime breakdown. \sysname averages 1.2 min for generation and 13.4 min for concolic execution, saving 5 min over \tool{RENGAR} despite slightly higher generation overhead.
\begin{table}[h!]
\centering
\caption{Runtime breakdown per project.}
\label{tab:time_breakdown}
\scriptsize
\setlength{\tabcolsep}{4pt}
\renewcommand{\arraystretch}{0.95}
\begin{tabular}{p{2.0cm}|
  >{\centering\arraybackslash}p{0.7cm}
  >{\centering\arraybackslash}p{0.7cm}
  >{\centering\arraybackslash}p{0.7cm}|
  >{\centering\arraybackslash}p{0.7cm}
  >{\centering\arraybackslash}p{0.7cm}
  >{\centering\arraybackslash}p{0.7cm}}
\toprule
\textbf{Target Project}
  & \multicolumn{3}{c|}{\textbf{\sysname}}
  & \multicolumn{3}{c}{\textbf{RENGAR}} \\
\cmidrule(lr){2-4}\cmidrule(lr){5-7}
  & \textbf{Total}
  & \textbf{Gen.}
  & \textbf{Con.}
  & \textbf{Total}
  & \textbf{Gen.}
  & \textbf{Con.} \\
\midrule
nltk                & 23.0 & 1.3 & 21.7 & 28.3 & 0.6 & 27.7 \\
python-markdown     & 4.0  & 0.9 & 3.1  & 5.2  & 0.3 & 4.9  \\
pygments            & 14.6 & 2.7 & 11.9 & 21.4 & 1.1 & 20.3 \\
sphinx              & 19.1 & 0.6 & 18.5 & 25.4 & 0.3 & 25.1 \\
peewee              & 9.3  & 0.4 & 8.9  & 12.8 & 0.2 & 12.6 \\
readme-ai           & 4.0  & 1.8 & 2.2  & 6.1  & 0.7 & 5.4  \\
numpy               & 41.9 & 2.3 & 39.6 & 47.4 & 0.9 & 46.5 \\
h5py                & 12.3 & 0.5 & 11.8 & 16.3 & 0.2 & 16.1 \\
langchain           & 30.9 & 1.6 & 29.3 & 35.9 & 0.7 & 35.2 \\
neural-compressor   & 5.1  & 0.2 & 4.9  & 9.3  & 0.1 & 9.2  \\
pydantic            & 7.8  & 1.1 & 6.7  & 13.0 & 0.4 & 12.6 \\
cffi                & 3.2  & 0.4 & 2.8  & 6.3  & 0.2 & 6.1  \\
\midrule
\textbf{Avg.} & \textbf{14.6} & \textbf{1.2} & \textbf{13.4} & \textbf{18.9} & \textbf{0.5} & \textbf{18.4} \\
\bottomrule
\end{tabular}\\[3pt]
{\footnotesize 1. Time in minutes. 2. Con.: concolic execution.}
\end{table}

\subsection{Concolic Execution Example}
\label{appendix:concolic_example}

We illustrate the concolic refinement of \sysname{} using the motivating example from Section~\ref{sec:background}.

\smallskip\noindent\textbf{Phase 1: Path Extraction.}
Given the entry point \texttt{is\_xml(text)}, \sysname constructs a static call graph and traverses backward from \texttt{tag\_re.search} (Line~10) to the entry.
Along this path, \texttt{check\_doctype}, \texttt{check\_xml} and \texttt{hash} are classified as $f_{\mathit{str}}$.

\smallskip\noindent\textbf{Phase 2: Function Summarization.}
\sysname{} performs concolic unit testing on \texttt{check\_doctype}, \texttt{check\_xml}, and \texttt{hash}.
Taking \texttt{check\_doctype} as an example, the four hypotheses are evaluated against $s_a$ in order:
\begin{enumerate}[leftmargin=*]
  \item $s_a \subseteq \mathit{pstr} \wedge \mathit{ret} = \text{const}$:
        \emph{refuted}; the function does not always return a constant.
  \item $s_a \subseteq \mathit{pstr} \wedge \mathit{ret} = s_a$:
        \emph{refuted}; the return value may differ when a DOCTYPE token is present.
  \item $ s_a = \mathit{pstr} \wedge \mathit{ret} = s_a$:
        \emph{validated}; when $s_a$ contains no DOCTYPE prefix,
        the function returns its input unchanged.
\end{enumerate}
Hypothesis~(iii) is adopted as the ReDoS-specific function summary:
  $\varphi_f:\quad \mathit{pstr} = s_a \;\Rightarrow\; \mathit{ret} = s_a$

\smallskip\noindent\textbf{Phase 3: System-Level Concolic Execution.}
\sysname{} composes the function summaries derived in Phase~2 and performs system-level concolic execution from \texttt{is\_xml}.
Each summarized function call is replaced by its corresponding $\varphi_f$: $\varphi_{\texttt{check\_xml}}$ at Line~3, $\varphi_{\texttt{hash}}$ at Line~5, and $\varphi_{\texttt{check\_doctype}}$ at Line~9.
\sysname{} then instruments the program with the negated assertion that $s_a$ does not reach \texttt{tag\_re.search} and invokes the SMT solver to find a counterexample.
The solver reasons over the following residual constraints without symbolically executing any summarized function:
\begin{itemize}[leftmargin=*]
  \item $\texttt{xml\_decl\_re.match}(s_a) = None $, to bypass the guard at Line~3.
  \item $s_a \notin \texttt{\_looks\_like\_xml\_cache}$, to bypass the cache at Line~7.
  \item $\mathit{pstr} = s_a \Rightarrow \mathit{ret} = s_a$, to propagate $s_a$ through \texttt{check\_doctype} at Line~9.
\end{itemize}
The solver produces a counterexample $s_a'$ satisfying all three constraints, confirming that $s_a'$ reaches \texttt{tag\_re.search} and is reported as a refined, exploitable attack string.

\newpage
\section{Meta-Review}

The following meta-review was prepared by the program committee for the 2026
IEEE Symposium on Security and Privacy (S\&P) as part of the review process as
detailed in the call for papers.

\subsection{Summary}
The proposed tool PUFFERDOS generates length-bounded ReDoS attack strings for backtracking regex engines. Experiments show notable efficiency gains over RENGAR and report reproducing most known CVEs plus discovering new exploitable instances.

\subsection{Scientific Contributions}
\begin{itemize}
\item Creates a New Tool to Enable Future Science.
\item Provides a Valuable Step Forward in an Established Field.
\end{itemize}

\subsection{Reasons for Acceptance}
\begin{enumerate}
\item Much shorter attack strings and large speedups compared to the baseline, improving exploit deliverability under input limits.
\item Useful integration of ReDoS-specific compositional concolic execution to validate exploitability in context (reachability/input constraints), reducing false positives.
\item Evidence of real-world impact via reproduction of known CVEs and disclosure of newly found cases.
\end{enumerate}

\subsection{Noteworthy Concerns}
\begin{enumerate}
\item The ``attack strings are too long'' framing is not fully substantiated with direct evidence about attacker feasibility/developer behavior, lacking clearer, concrete threat-model justification and examples.
\item Analysis relies on a simplified regex model and a fixed pattern taxonomy, potentially missing ReDoS cases requiring advanced regex features.
\item Comparison is primarily against one baseline only (RENGAR); broader SOTA comparisons are missing.
\end{enumerate}

%% file: Tex/Table/pattern.tex
\begin{table}[h!]
\centering
\caption{Vulnerable patterns, examples, and coverage.}
\begin{tabular}{c|c|c}
\toprule[1.5pt]
\textbf{Pattern} & \textbf{Example} & \textbf{Coverage} \\ \midrule
EOL & \texttt{\string^(ab|a|b)+\$} & EOLS~\cite{wang2023effective}, EOA, NQ, EOD~\cite{li2021redoshunter} \\
LIL & \verb|.*a.*|               & PTLS,POLS~\cite{wang2023effective}, 
POA,SLQ~\cite{li2021redoshunter} \\
PML & \verb|\s*a\s*b\s+|            & PTLS~\cite{wang2023effective}, POA~\cite{li2021redoshunter} \\
\bottomrule[1.5pt]
\end{tabular}
\label{tab:regex_patterns}
\end{table}

%% file: Tex/Table/realworld.tex
\begin{table}[]
\caption{Project statistics}
\label{tab:project}
\centering
\begin{tabular}{p{2.2cm}|
  >{\centering\arraybackslash}p{0.4cm}
  >{\centering\arraybackslash}p{0.78cm}|
  >{\centering\arraybackslash}p{0.5cm}
  >{\centering\arraybackslash}p{0.5cm}
  >{\centering\arraybackslash}p{0.8cm}
  >{\centering\arraybackslash}p{0.75cm}}
\midrule
\textbf{Project} 
  & \textbf{Stars} 
  & \textbf{LOC} 
  & \textbf{Func} 
  & \textbf{\#$f$} 
  & \textbf{$\phi_f$}
  & \textbf{Vul. $r$} \\
\midrule
nltk              & 14.4k & 119{,}361 & 6147   & 678 & 189  & 8 \\
python-markdown   & 2.8k  & 7,860  & 391    & 170  & 80   & 9 \\
pygments          & 2.1k  & 122{,}102 & 1353   & 611  & 188  & 15 \\
sphinx            & 7.5k  & 121{,}070 & 7142   & 1140 & 274  & 10 \\
peewee            & 11.8k & 40{,}444  & 3508   & 680  & 161  & 1 \\
readme-ai         & 2.7k  & 10{,}332  & 585    & 282  & 32  & 8 \\
numpy             & 30.7k & 235{,}972 & 13{,}103 & 1089 & 446 & 21 \\
h5py              & 2.2k  & 15{,}923  & 1328   & 485  & 175  & 3 \\
langchain         & 119k  & 250{,}993 & 9632   & 474 & 209 & 14 \\
neural-compressor & 2.5k  & 443{,}427 & 13{,}530 & 316 & 43 & 1 \\
pydantic          & 25.6k & 91{,}150  & 5927   & 291 & 54  & 4 \\
cffi              & 221  & 31{,}112  & 2373   & 632  & 44  & 1\\ 
\midrule
\textbf{Total} & -- & -- & \textbf{68,021} & \textbf{5,848} & \textbf{1,895} & \textbf{95} \\
\bottomrule
\end{tabular}
\end{table}

%% file: Tex/Table/diff_engine.tex
\begin{table}[t!]
\centering
\caption{Cross-Engine evaluation results}
\label{tab:cross_engine_eval}
\renewcommand{\arraystretch}{0.9}
\begin{tabular}{
>{\centering\arraybackslash}m{1.4cm}|
>{\centering\arraybackslash}m{0.7cm}>{\centering\arraybackslash}m{0.6cm}|
>{\centering\arraybackslash}m{0.7cm}>{\centering\arraybackslash}m{0.6cm}|
>{\centering\arraybackslash}m{0.7cm}>{\centering\arraybackslash}m{0.6cm}
}
\toprule[1.5pt]
\textbf{Engine} 
& \multicolumn{2}{c|}{$T = 0.1$} 
& \multicolumn{2}{c|}{$T = 1$} 
& \multicolumn{2}{c}{$T = 10$} \\
\cmidrule(lr){2-3} \cmidrule(lr){4-5} \cmidrule(lr){6-7}
& \textbf{$\bar{R}_T$} & \textbf{$P_T$}
& \textbf{$\bar{R}_T$} & \textbf{$P_T$}
& \textbf{$\bar{R}_T$} & \textbf{$P_T$} \\
\midrule
Python-3.12 & 98.8 & 88.0\% & 928.5 & 89.9\% & 2751.5 & 91.9\% \\
Node.js-25  & 12.1 & 73.1\% & 69.3 & 82.9\% & 99.2 & 83.9\% \\
Java-8      & 185.4 & 92.5\% & 2042.7 & 92.1\% & 5051.8 & 93.2\% \\
Java-23     & 28.1 & 71.4\% & 162.9 & 73.1\% & 318.7 & 73.8\% \\
\bottomrule[1.5pt]
\end{tabular}
\end{table}

\vspace{1.2em}

\begin{table}[!t]
\centering
\caption{Pattern-Based cross-engine evaluation results}
\label{tab:pattern_cross_engine}
\renewcommand{\arraystretch}{1.1}

\begin{tabular}{
>{\centering\arraybackslash}m{1.4cm}|
>{\centering\arraybackslash}m{0.7cm}>{\centering\arraybackslash}m{0.6cm}|
>{\centering\arraybackslash}m{0.7cm}>{\centering\arraybackslash}m{0.6cm}|
>{\centering\arraybackslash}m{0.7cm}>{\centering\arraybackslash}m{0.6cm}
}
\toprule[1.5pt]
\multicolumn{7}{c}{\textbf{(a) EOL}} \\ 
\midrule
\textbf{Engine} 
& \multicolumn{2}{c|}{$T = 0.1$} 
& \multicolumn{2}{c|}{$T = 1$} 
& \multicolumn{2}{c}{$T = 10$} \\
\cmidrule(lr){2-3} \cmidrule(lr){4-5} \cmidrule(lr){6-7}
& \textbf{$\bar{R}_T$} & \textbf{$P_T$}
& \textbf{$\bar{R}_T$} & \textbf{$P_T$}
& \textbf{$\bar{R}_T$} & \textbf{$P_T$} \\
\midrule
Python-3.12 & 0.01 & 100\% & 0.01 & 100\% & 0.01 & 100\% \\
Node.js-25  & 0.00 & 0\% & 0.00 & 0\% & 0.00 & 0\% \\
Java-8      & 0.05 & 100\% & 0.10 & 100\% & 0.10 & 100\% \\
Java-23     & 0.00 & 0\% & 0.01 & 100\% & 0.01 & 100\% \\
\bottomrule[0.5pt]
\end{tabular}

\vspace{0.2em}

\begin{tabular}{
>{\centering\arraybackslash}m{1.4cm}|
>{\centering\arraybackslash}m{0.7cm}>{\centering\arraybackslash}m{0.6cm}|
>{\centering\arraybackslash}m{0.7cm}>{\centering\arraybackslash}m{0.6cm}|
>{\centering\arraybackslash}m{0.7cm}>{\centering\arraybackslash}m{0.6cm}
}
\toprule[1.5pt]
\multicolumn{7}{c}{\textbf{(b) LIL}} \\ 
\midrule
\textbf{Engine} 
& \multicolumn{2}{c|}{$T = 0.1$} 
& \multicolumn{2}{c|}{$T = 1$} 
& \multicolumn{2}{c}{$T = 10$} \\
\cmidrule(lr){2-3} \cmidrule(lr){4-5} \cmidrule(lr){6-7}
& \textbf{$\bar{R}_T$} & \textbf{$P_T$}
& \textbf{$\bar{R}_T$} & \textbf{$P_T$}
& \textbf{$\bar{R}_T$} & \textbf{$P_T$} \\
\midrule
Python-3.12 & 101.3 & 91.2\% & 1021.2 & 91.1\% & 3081.8 & 91.7\% \\
Node.js-25  & 14.9 & 74.4\% & 112.9 & 83.8\% & 175.5 & 88.1\% \\
Java-8      & 245.2 & 93.6\% & 2359.8 & 93.0\% & 5510.1 & 93.2\% \\
Java-23     & 38.7 & 74.5\% & 191.6 & 74.0\% & 472.3 & 80.8\% \\
\bottomrule[0.5pt]
\end{tabular}

\vspace{0.2em}

\begin{tabular}{
>{\centering\arraybackslash}m{1.4cm}|
>{\centering\arraybackslash}m{0.7cm}>{\centering\arraybackslash}m{0.6cm}|
>{\centering\arraybackslash}m{0.7cm}>{\centering\arraybackslash}m{0.6cm}|
>{\centering\arraybackslash}m{0.7cm}>{\centering\arraybackslash}m{0.6cm}
}
\toprule[1.5pt]
\multicolumn{7}{c}{\textbf{(c) PML}} \\ 
\midrule
\textbf{Engine} 
& \multicolumn{2}{c|}{$T = 0.1$} 
& \multicolumn{2}{c|}{$T = 1$} 
& \multicolumn{2}{c}{$T = 10$} \\
\cmidrule(lr){2-3} \cmidrule(lr){4-5} \cmidrule(lr){6-7}
& \textbf{$\bar{R}_T$} & \textbf{$P_T$}
& \textbf{$\bar{R}_T$} & \textbf{$P_T$}
& \textbf{$\bar{R}_T$} & \textbf{$P_T$} \\
\midrule
Python-3.12 & 0.9 & 100\% & 1.0 & 100\% & 1.1 & 100\% \\
Node.js-25  & 0.2 & 100\% & 0.2 & 100\% & 0.2 & 100\% \\
Java-8      & 1.8 & 100\% & 1.4 & 100\% & 1.6 & 100\% \\
Java-23     & 0.2 & 100\% & 0.3 & 100\% & 0.5 & 100\% \\
\bottomrule[0.5pt]
\end{tabular}
\end{table}